\documentclass[pra,8pt,twocolumn,superscriptaddress,floatfix,showpacs]{revtex4-2}
\usepackage{tikz}
\usetikzlibrary{arrows.meta, positioning,shapes,calc}
\usepackage[utf8]{inputenc}
\usepackage[T1]{fontenc}     
\usepackage[british]{babel}  
\usepackage[sc,osf]{mathpazo}
\usepackage[scaled=0.86]{berasans}  
\usepackage[colorlinks=true, citecolor=blue, urlcolor=blue]{hyperref}  
\usepackage{graphicx} 
\usepackage{subfig}
\usepackage[babel]{microtype}  
\usepackage{amsmath,amssymb,amsthm,bm,amsfonts,mathrsfs,bbm} 
\usepackage{comment}

\usepackage{xspace}  
\usepackage{pgfplots}
\usepackage{xcolor,colortbl}
\def\ba{\begin{equation}}
	\def\ea{\end{equation}}
\def\bea{\begin{eqnarray}}
	\def\eea{\end{eqnarray}}
\def\ben{\begin{equation*}}
	\def\een{\end{equation*}}
\def\bean{\begin{eqnarray*}}
	\def\eean{\end{eqnarray*}}
\def\bma{\begin{mathletters}}
	\def\ema{\end{mathletters}}
\def\bi{\begin{itemize}}
	\def\ei{\end{itemize}}

\newcommand{\be}{\begin{equation}}
	\newcommand{\ee}{\end{equation}}

\newcommand{\kommentar}[1]{}

\newcommand{\forget}[1]{}

\newtheorem{theorem}{Theorem}

\newtheorem{corollary}{Corollary}[theorem]

\begin{document}
	
	\title{
		Exploiting Full Network Nonlocality From Quantum Key Distribution Perspective}
	\author{Kaushiki Mukherjee}
	\email{kaushiki.wbes@gmail.com}
	\affiliation{Department of Mathematics, Jhargram Raj College,Jhargram, West Bengal 721507, India.}
	
	
	\begin{abstract}
		Recently, the notion of full network nonlocality has emerged as some truly quantum network phenomena that cannot be realized in traditional Bell experiments. Present work manifests the utility of such form of truly network non-classicality in designing a four partite network-based entanglement assisted quantum key distribution protocol. Precisely, validity of the QKD protocol relies upon full network nonlocality detection via violation of some suitable trilocal inequality.
		 Based on the quantum bit error rate and violation of trilocal inequality, arbitrary two qubit entangled states are characterized in accordance with their utility in successfully executing the protocol. The protocol is further analyzed from a restricted Eve-aware perspective by quantifying the Eve's partial side information and discussing its implications for the corresponding Devetak–Winter $1$-way secret-key rate. Also, considering Eve's control over the trusted parties' measurement devices, an illustration of the protocol's resistance to detection loophole is provided. Due to connected structure of entangled sources, any genuine form of network nonlocality may offer advantage over standard Bell nonlocality for designing secure key distribution protocols. To establish such an intuition as a fact, another QKD protocol, relying only upon Bell-CHSH violation, is designed. The former offers more stricter validity checks compared to the latter. Importantly, while the quantum bit error rate can be less than $14.6\%$ exploiting Bell-CHSH nonlocality, it can be reduced below $13.7\%$ by exploiting full network nonlocality. 
	\end{abstract}
	\maketitle
	
	
	\section{Introduction}\label{intro}
	Basic task of a key distribution protocol is to generate a private key between trusted parties who are allowed to communicate among themselves over public channels\cite{retsispg10}. Speaking of secure key generation, protocols that can avail quantum resources largely outperform protocols which only involve classical resources\cite{wies,QC1,QC2}.Such advantage offered by quantum key distribution protocols basically stems from the intrinsic randomness of quantum particles\cite{bb84,eck,cr,cr1,cr2,cr3,cr4,cr5,cr6}. 
\par Entire class of QKD protocols can be broadly classified in two types: \textit{preparation and measurement type}\cite{bb84,cr,cr1,cr2,cr3,cr4,cr5,cr6} and \textit{entanglement assisted type}l\cite{eck,eck3,eck1,eck2,eck4,eck6,eck5,eck7}. Ekert protocol\cite{eck} was the first entanglement assisted type of QKD protocol. Since Ekert's protocol, several QKD
protocols have been introduced following Ekert's protocol\cite{eck3,eck1,eck2,eck4,eck6,eck5,eck7}. Exploitation of non-classical resource in the form of quantum entanglement lies at the root of this type of key distribution protocols. In any entanglement assisted key distribution protocols\cite{eck} the trusted parties($A$ and $B$,say) share an entangled state\cite{hor1}. Precisely, multiple copies of an entangled state are distributed in between $A$ and $B$. They perform some local measurements in mutually unbiased bases(MUBs\cite{mub1}) on their respective subsystems. They use a fraction of these outcomes of local measurements in MUBs for validation check. Remaining outcomes are used for public reconciliation of their measurement bases and formation of raw key\cite{bran1} from the outcomes of identical MUBs. Raw key is then used to extract a secure key via suitable information reconciliation\cite{eck8i} and privacy amplification\cite{eck8i} strategies.	
	\par Now, in the protocol, communication being made over public quantum channel, any dishonest third party(Eve,say) can intercept, measure the qubits(sent to $A$ and $B$) and then send some new qubits to them thereby hampering the security of the key generation. Presence of eavesdroppers can be detected by $A$ and $B$ who then abort the protocol. Such utility of Bell inequality violation for detecting presence of any untrusted party was first pointed out by Ekert\cite{eck}. Since then such Bell inequality's violation based verification strategy has been exploited to design secure key generation protocols\cite{eck1,eck2,eck3,eck4,eck5,eck6,eck7,main}. Many of these research works have also exploited Bell-type inequality violation for minimizing Quantum Bit Error Rate(QBER) generated in the QKD protocols\cite{main,selfq}   \\
	Bell inequality violation turns out to be a necessary(though not sufficient\cite{main}) criterion for checking security in QKD protocol. Till date, all such verification schemes rely on standard Bell measurement scenario involving a single quantum source. However, with development of quantum technology, study of quantum information science has moved beyond regime of single source frameworks thereby witnessing exploitation of network structure involving more than one quantum source. In this context, one obvious query arises: \textit{can a network framework involving multiple quantum sources offer better validation check in key distribution?} Present work will explore in this direction.
	\par That source independence assumption reduces requirements to demonstrate non-classicality of quantum correlations compared to usual Bell-CHSH scenario, was first pointed out in \cite{BGP}. Framework of correlations in bilocal scenario was then formalized in \cite{BRGP}. A series of works then followed exploiting several intriguing features of quantum correlations in network under source independence assumption\cite{frtz1,gis1,Tava,km1,km2,bilo2,lee,bilo3,bilo4,km3,nr4,km4,bilo7,bilo5,km5,nr2,nr1,ejm,bilo6,nr3,birev,km7,kau6,nr5,km8,bilo11}. In past few years, notion of full network nonlocality has been introduced\cite{bilo6}. Such form of genuine nonlocal network correlations necessitates distribution of nonlocal resources by all links in a network. Manifestation of full network nonlocality in $n$-local networks\cite{bilo6}for designing an entanglement assisted QKD protocol will form the mainstay of present exploration.
\par	Till date, to the best of author's knowledge, limited effort has been given to design $n$-local network based QKD protocols\cite{lee}. In \cite{lee} the author introduced QKD protocols using such type of networks. Using DAG approach, the authors exploited standard network nonlocality to build more secure key distribution protocols\cite{lee} compared to standard Bell scenario based protocols. Observations, pertaining to security analysis, made therein relied upon violation of $n$-local inequalities(denoting collection of such inequalities as $\mathcal{I},$say). However, in recent times it has been argued that conceptually standard network nonlocality is neither novel(compared to standard Bell nonlocality) nor truly a network phenomenon\cite{birev,bilo6}. In \cite{bilo6}, the authors pointed out that violation of existing $n$-local inequalities($\mathcal{I}$) can be obtained even if all the bipartite sources are not nonlocal. Consequently such violation cannot be attributable to intrinsic structure of network(multiple sources framework). In this context, a new notion of network nonlocality, referred to as \textit{full network nonlocality(FNN)} was introduced in \cite{bilo6}. Such form of non-classicality turns out to be a truly network phenomenon. In this work, full network nonlocality will be exploited for the purpose of security analysis in a QKD protocol.  \\ 
	A network based QKD protocol involving four trusted parties is designed here. The network underlying the protocol is a star-shaped $4$-local network\cite{bilo6} with a central and three extreme parties. Security steps in such a protocol exploit violation of an existing trilocal inequality detecting full network nonlocality. Threshold value of QBER is derived in absence of violation of the trilocal inequality. However, on observing violation, bit error rate can be diminished below the critical value. Such results aid in characterizing arbitrary two qubit states for designing the protocol.\\
	In literature violation of Bell-CHSH inequality\cite{cla} has been the mainspring for secure key generation in an entanglement assisted QKD protocols. In this context, it becomes imperative to explore whether security in above network based protocol can be offered by Bell-CHSH inequality's violation. For this purpose, another four-partite network based protocol is designed which is similar to the former one. However, here security is provided by Bell-CHSH violation in each of three pairs of central and an extreme party. Interestingly, it turns out that violation of trilocal inequality helps in framing more stringent security criteria compared to that depending upon Bell-CHSH violations only.  \\
	As already said, the protocol designed here is four-partite and involves only two-qubit entangled states. To this end it may be noted that in \cite{pan}, the authors have put forward a set of arbitrary-party and unbounded-input network inequalities in star-shaped and linear-chain networks whose optimal quantum violation exhibits full network nonlocality. Analyzing such set of generalized inequalities for key distribution lies beyond the scope of present discussion.\\
	Rest of the work is organised as follows: in sec.\ref{pre} basic preliminaries are provided. First network based protocol is given in sec.\ref{trilov} followed by characterization of states in sec.\ref{states}. An eve-aware analysis of the protocol is provided in sec.\ref{new1} followed by analysis of the detection efficiency of the protocol in sec.\ref{new2}. Second protocol is designed in sec.\ref{cpmp}. Finally  some concluding remarks are provided in sec.\ref{conc}.
	\section{Preliminaries}\label{pre}
	\subsection{Bloch Matrix Representation}
	The density matrix of an arbitrary bipartite two qubit state($\rho$) is given by\cite{luo,gam}:
	\begin{equation}\label{st4}
		\small{\rho}=\small{\frac{1}{4}(\mathbb{I}_{2}\times\mathbb{I}_2+\vec{\mathbf{a}}.\vec{\sigma}\otimes \mathbb{I}_2+\mathbb{I}_2\otimes \vec{\mathbf{b}}.\vec{\sigma}+\sum_{i_1,i_2=1}^{3}r_{i_1i_2}\sigma_{i_1}\otimes\sigma_{i_2})},
	\end{equation}
	with $\vec{\sigma}$$=$$(\sigma_1,\sigma_2,\sigma_3), $ $\sigma_{j_k}$ denoting Pauli operators inclined along $3$ mutually perpendicular directions($i_k$$=$$1,2,3$). $\vec{\mathbf{a}}$$=$$(a_1,a_2,a_3)$ and $\vec{\mathbf{b}}$$=$$(b_1,b_2,b_3)$ denote local bloch vectors($\vec{\mathbf{a}},\vec{\mathbf{b}}$$\in$$\mathbb{R}^3$) corresponding to party $A$ and $B$ respectively with $|\vec{\mathbf{a}}|,|\vec{\mathbf{b}}|$$\leq$$1$ and $(r_{i,j})_{3\times3}$ denotes the correlation tensor matrix $\mathcal{R}$(real matrix).\\
	Components $r_{j_1j_2}$ of $\mathcal{R}$ are given by $r_{j_1j_2}$$=$$\textmd{Tr}[\rho\,\sigma_{j_1}\otimes\sigma_{j_2}].$ \\
	$\mathcal{R}$ can be diagonalized by applying suitable local unitary operations\cite{gam,luo},where the simplified expression is then given by:
	\begin{equation}\label{st41}
		\small{\rho}^{'}=\small{\frac{1}{4}(\mathbb{I}_{2}\times\mathbb{I}_2+\vec{\mathbf{m}}.\vec{\sigma}\otimes \mathbb{I}_2+\mathbb{I}_2\otimes \vec{\mathbf{n}}.\vec{\sigma}+\sum_{i=1}^{3}t_{i}\sigma_{i}\otimes\sigma_{i})}.
	\end{equation}
	Correlation tensor in Eq.(\ref{st41}) is given by $T$$=$$\textmd{diag}(t_{1},t_{2},t_{3})$ where $t_{1},t_{2},t_{3}$ are the eigen values of $\sqrt{\mathcal{R}^{T}\mathcal{R}},$ i.e., singular values of $\mathcal{R}.$
	\subsection{Entanglement Assisted Bipartite QKD Protocol}\label{qkd}
	Consider an entanglement assisted quantum key distribution(QKD) protocol\cite{eck,main} that involves two trusted parties $A$ and $B.$ At the end of the protocol they try to establish a secure key. Let $A$ prepare several copies of a two-qubit state($\rho$) and send one qubit of each such copies of $\rho$ to $B$. After distribution of qubits, both of them perform local measurements on their respective qubits. For local measurements, each of them selects randomly from a collection of $n$ number of $d$-dimensional mutually unbiased bases(MUBs). Let  $\mathfrak{C}_{1}$$=$$\{\mathfrak{B}_{1}^{(\alpha)}\}_{\alpha=1}^{n}$ denote the collection of $n$ MUBs from which each of $A$ and $B$ chooses randomly. $\forall\,\alpha,$ $\mathfrak{B}_{1}^{(\alpha)}$ are given by:
	\begin{equation}\label{basis1}
		\mathfrak{B}_{1}^{(\alpha)}=\{|\psi^{\alpha}_i\rangle\}_{i=1}^d
	\end{equation}
	If $\mathcal{O}_{1}^{(\alpha)}$ denote operators corresponding to the basis $\mathfrak{B}_{1}^{(\alpha)}$, then those are given by:
	\begin{equation}\label{opes}
		\mathcal{O}_{1}^{(\alpha)}=\{|\psi^{\alpha}_i\rangle\langle \psi^{\alpha}_i|\}_{i=1}^d,\,\forall \alpha=1,...,n.
	\end{equation}
	After performing measurements on $n$ copies of $\rho,$ $A$ and $B$ use a fraction of the measurement outcomes to verify whether corresponding correlations are nonlocal by testing violation of a suitable Bell inequality. For the remaining part of the measurement outcomes, the trusted parties publicly compare their measurement bases and keep outcomes only corresponding to the identical bases while discard the remaining outcomes. The outcomes obtained from identical bases form the sifting key\cite{bran1}. They use a part of this key, i.e., measurement outcomes when their bases are same, to calculate quantum bit error rate(QBER). If QBER is less than some preset critical value, it is used to extract secure key by information reconciliation\cite{eck8i} and privacy amplification\cite{eck8i}.
	\subsection{Quantum Bit Error Rate}\label{qber}
	For any given state $\varrho,$ QBER($Q$) is considered as the average mismatch between $A$ and $B$'s outcomes obtained when they measure in identical bases. With $\mathfrak{C}_{1}$ denoting collection of $n$ MUBs(Eq.(\ref{basis1})) from which each of the two parties selects randomly, QBER can be expressed as:
	\begin{equation}\label{basis2}
		Q=\frac{1}{n}\sum_{\beta=1}^{n}\sum_{i\neq j=1}^d \langle \psi^{\beta}_i\psi^{\beta}_j|\rho|\psi^{\beta}_i\psi^{\beta}_j\rangle
	\end{equation}
	The above expression of $Q$ holds for any $n$$\leq$$d+1$ number of bases. For instance, when $\rho,$ shared between $A$ and $B$ is a two qubit state($d$$=$$2$) and each party chooses from a collection of two bases, i.e., $|\mathfrak{C}_{1}|$$=$$2,$ QBER is given by\cite{main}:
	\begin{equation}\label{basis3}
		Q=\frac{1}{4}(2-\vec{u_1}.T\vec{u_1}-\vec{u_2}.T\vec{u_2})
	\end{equation}
	where $\vec{u}_i\,(i$$=$$1,2)$ denote Bloch vectors of the measurement bases and $T$ denotes the correlation tensor(Eq.(\ref{st4})). Minimization over all possible measurement directions $\vec{u}_1,\vec{u}_2$ gives:
	\begin{equation}\label{basis4}
		Q\geq \frac{1}{4}(2-\textmd{max}_{i,j}(|t_{i}|+|t_{j}|)),\,i\neq j
	\end{equation}
	where $t_{1},t_{2},t_{3}$ denote the singular values of correlation tensor $T$ of $\rho^{'}$(Eq.(\ref{st41})) and hence singular values of correlation tensor $T$ of $\rho$(Eq.(\ref{st4})).
	\subsection{$n$-local Star Network}
	$n$-local star network($\mathcal{N}_{n-star},$say) is a non-linear network connecting a single central party $A_1$ to $n$ edge(extreme) parties $A_2,A_3,...,A_{n+1}$ (see Fig.\ref{star}) \\
	$A_1$receives one particle from each source $\mathbf{S}_1,\mathbf{S}_2,...,\mathbf{S}_n$  whereas each of the extreme parties($A_{i+1}$) receives one particle from one source($\mathbf{S}_i,i$$=$$1,2,...,n$). $\forall i,\mathbf{S}_i$ is characterized by variable $\lambda_i(i$$=$$1,2,...,n).$ As $\mathbf{S}_1,\mathbf{S}_2,...,\mathbf{S}_n$ are independent of each other, joint distribution of the variables $\lambda_1,\lambda_2,...,\lambda_n$ is factorizable:
	\begin{equation}\label{tr1}
		\mathbf{p}(\lambda_1,\lambda_2,...,\lambda_n)=\Pi_{i=1}^n \mathbf{p}_i(\lambda_i),
	\end{equation}
	where $\mathbf{p}_i$ denotes the normalized distribution of $\lambda_i,\forall i.$ Source independence condition(Eq.(\ref{tr1})) represents the $n$-local constraint\cite{BRGP}. \\
	Each of the edge parties $A_i$ chooses to perform from any one of two possible dichotomic measurements($x_i$$\in$$\{0,1\}$) whereas the central party performs a fixed measurement. Let $\bar{a}_1$ denote output bit string resulting from single measurement of $A_1$ and let $a_i$$\in$$\{0,1\}$ denote outcomes of $x_i\forall i.$ Let $p(\bar{a}_1,a_2,...,a_{n+1}|x_2,...,x_{n+1})$ denote corresponding measurement correlation term.
	\begin{figure}
		\begin{tikzpicture}[
			node distance=1.4cm and 1.6cm,
			every node/.style={font=\boldmath\bfseries\small},
			meas/.style={
				draw, circle, minimum size=0.6cm,
				inner sep=0pt, align=center,
				text height=1.2ex, text depth=0ex
			},
			source/.style={
				draw, rectangle, minimum width=0.5cm, minimum height=0.3cm,
				inner sep=1pt, align=center
			},
			arr/.style={->, thick}
			]
			
			\node[meas] (A1) at (0,0) {$A_1$};
			
			\node[meas] (A2) at (150:3cm) {$A_2$};
			\node[meas] (A3) at (120:3cm) {$A_3$};
			
			\node[meas] (An) at (60:3cm) {$A_n$};
			\node[meas] (Anp1) at (30:3cm) {$A_{n+1}$};
			
			\node at (90:3.2cm) {\dots};
			
			\node[source] (S1) at ($(A1)!0.5!(A2)$) {$\mathbf{S}_1$};
			\node[source] (S2) at ($(A1)!0.5!(A3)$) {$\mathbf{S}_2$};
			\node[source] (Sn) at ($(A1)!0.5!(An)$) {$\mathbf{S}_{n-1}$};
			\node[source] (Snp1) at ($(A1)!0.5!(Anp1)$) {$\mathbf{S}_n$};
			
			\foreach \X in {S1, S2, Sn, Snp1} {
				\draw[arr] (\X) -- (A1);
			}
			\draw[arr] (S1) -- (A2);
			\draw[arr] (S2) -- (A3);
			\draw[arr] (Sn) -- (An);
			\draw[arr] (Snp1) -- (Anp1);
			
			\foreach \i/\A in {2/A2, 3/A3, n/An, {n+1}/Anp1} {
				\node at ($(\A)+(0,0.9)$) {$x_{\i}$};
				\draw[arr] ($(\A)+(0,0.7)$) -- (\A);
				\draw[arr] (\A) -- ++(0,-0.8) node[below] {$a_{\i}$};
			}
			
			\draw[arr] (A1) -- ++(0,-0.6) node[below] {$\bar{a}_1$};
		\end{tikzpicture}
		\caption{Schematic Diagram of $n$-local star network}
		\label{star}
	\end{figure}
	\subsection{Full Network Nonlocality(FNN)}
	In any given measurement scenario, where all the sources are independent(Eq.(\ref{tr1})), network correlations are said to be \textit{fully network nonlocal}\cite{bilo6} if and only if one cannot model the correlations by a hidden variable(HV) model such that at least one source in the network is of a local-variable nature whereas all the remaining sources, in general, can be independent nonlocal resources.\\
	For instance, as required for present work, consider the measurement scenario corresponding to $n$-local star network($\mathcal{N}_{n-star}$).\\
	$p(\bar{a}_1,a_2,...,a_{n+1}|x_2,...,x_{n+1})$ is \textit{not fully network nonlocal} if it cannot be described by a model that allows at least one of $n$ independent sources in the network to be local whereas the others are in general nonlocal(even no-signalling) in nature. Mathematically the concept can be described as follows.\\
	Without loss of any generality, let $\mathbf{S}_1$ be local-variable in nature. So the hidden variable $\lambda_1$ corresponding to such source is a local hidden variable. Let all other sources be maximally nonlocal. Then each of $\mathbf{S}_2,\cdots,\mathbf{S}_n$ can in general be modeled by bipartite no-signalling box. Then $p(\bar{a}_1,a_2,...,a_{n+1}|x_2,...,x_{n+1})$ can decomposed as:
	\begin{eqnarray}\label{fnn1}
		\small{p(\bar{a}_1,a_2,...,a_{n+1}|x_2,...,x_{n+1})}=\sum_{\lambda_1}\mathbf{p}_1(\lambda_1)	\small{p(a_2|x_2,\lambda_1)}\mathbf{Q}\nonumber\\
		\\
		\textmd{\small{with}}\,\,\mathbf{Q}=\small{p(\bar{a}_1,a_3,\cdots,a_{n+1}|x_3,\cdots,x_{n})}\nonumber
	\end{eqnarray}
	If$\,p(\bar{a}_1,a_2,\cdots,a_{n+1}|x_2,...,x_{n+1})$ satisfies above form(Eq.(\ref{fnn1})) or any other analogous form considering any one of $\lambda_2,\cdots, \lambda_n$ to be modeled by local hidden variable then 
	$p(\bar{a}_1,a_2,...,a_{n+1}|x_2,...,x_{n+1})$ is \textit{not fully network nonlocal}.
	\subsection{Detection of FNN in Trilocal Star Network}
	In \cite{bilo6}, the authors gave a trilocal inequality whose violation indicates full network nonlocality of corresponding $4$-partite measurement correlations. The correlators based inequality\cite{bilo6} is given by:
	\begin{eqnarray}\label{ineqs}
		&& \frac{1}{2}\sum_{i=1}^{4}|J_{i}|^{\frac{1}{3}}\leq  2^{\frac{1}{3}},\,  \textmd{where}\\
		&& J_i=\frac{1}{2^3}\sum_{x_2,x_3,x_{4}}(-1)^{g_i(x_2,x_3,x_{4})}\langle A_{(1)}^{(i)}A_{x_2}A_{x_{3}}A_{x_{4}}\rangle\nonumber\\
		&&\langle A_{(1)}^{(i)}A_{x_2}A_{x_3}A_{x_{4}}\rangle=\sum_{\mathcal{D}_1}(-1)^{\tilde{\mathfrak{a}}_1^{(i)}+\textbf{a}_2+\textbf{a}_3+\textbf{a}_{4}}M_1,\nonumber\\
		&& \textmd{\small{where}}\,M_1=\small{p(\tilde{\mathfrak{a}}_1^{(1)},\tilde{\mathfrak{a}}_1^{(2)},\tilde{\mathfrak{a}}_1^{(3)},\tilde{\mathfrak{a}}_1^{(4)},\textbf{a}_2,\textbf{a}_3,\textbf{a}_{4}|x_2,x_3,x_4)}\,\textmd{\small{and}}\nonumber\\
		&&\mathcal{D}_1=\{\tilde{\mathfrak{a}}_1^{(1)},\tilde{\mathfrak{a}}_1^{(2)},\tilde{\mathfrak{a}}_1^{(3)},\tilde{\mathfrak{a}}_1^{(4)},\textbf{a}_2,\textbf{a}_3,\textbf{a}_{4}|\tilde{\mathfrak{a}}_1^{(i)},\textbf{a}_k\in\{0,1\}\}\nonumber\\
	\end{eqnarray}
	In Eq.(\ref{ineqs}), $\forall i$$=1,2,3,4,$ $\tilde{\mathfrak{a}}_1^{(i)}$ stands for an output bit generated by classical post-processing of the raw output string $\overline{\textbf{a}}_1$$=$$(\textbf{a}_{11},\textbf{a}_{12},\textbf{a}_{13})$ of $A_1.$ In Eq.(\ref{ineqs}), $\forall i,\, g_i$ are functions of the input variables $x_2,x_3,x_{4}$ of the extreme parties\cite{Tava}. Each $g_i$ contains an even number of $x_2,x_3,x_{4}.$ Classical post-processed bits $\tilde{\mathfrak{a}}_1^{(i)}$ from the output string $\overline{\textbf{a}_1}$ and corresponding functions $g_i(x_2,x_3,x_{4})$ are provided in Table.\ref{table:ta6}.
	\begin{table}
		\caption{The table gives classically post-processed bits $\tilde{a}_1^{(i)}$ and also $g_i(x_2,x_3,x_{4})$ appearing in Eq.(\ref{ineqs}).}
			\begin{tabular}{|c|c|}
				\hline
				$\tilde{a}_1^{(i)}$ &$g_i(x_2,x_3,x_{4})$  \\
				\hline
				$\tilde{a}_1^{(1)}$$=$$\textbf{a}_{11}$,
				$\tilde{a}_1^{(2)}$$=$$\textbf{a}_{11}$$\oplus$$\textbf{a}_{12}$$\oplus$$1$&$g_1$$=$$0,$
				$g_2$$=$$x_2$$+$$x_3$\\
				$\tilde{a}_1^{(3)}$$=$$\textbf{a}_{11}$$\oplus$$\textbf{a}_{13}$$\oplus$$1$
				&$g_3$$=$$x_2$$+$$x_4$\\
				$\tilde{a}_1^{(4)}$$=$$\textbf{a}_{11}$$\oplus$$\textbf{a}_{12}$$\oplus$$\textbf{a}_{13}
				$$\oplus$$1$&$g_4$$=$$x_3$$+$$x_4$\\
					\hline
			\end{tabular}
		\label{table:ta6}
	\end{table}
	Now considering a trilocal network where each source($\mathbf{S}_i$) distributes an arbitrary two-qubit state(Eq.\ref{st41}). Let $A_1$ measure joint state of three qubits in tripartite GHZ basis. Let each of the edge parties($A_{i+1}$) perform single qubit projective measurement in any one of two arbitrary directions:$\{\vec{x}_{i+1}.\vec{\sigma},\vec{x}^{'}_{i+1}.\vec{\sigma}\}$ Under these measurement settings, the upper bound($\mathcal{B}_{3-star},$say) of trilocal inequality(Eq.(\ref{ineqs})) is given by \cite{bilo5}: 
	\begin{equation}\label{boundstar}
		\mathcal{B}_{3-star}=	\sqrt{(\Pi_{i=1}^3  t_{i,1})^{\frac{2}{3}}+(\Pi_{i=1}^3t_{i,2})^{\frac{2}{3}}}.
	\end{equation}
	In Eq.(\ref{boundstar}) $t_{i,1}$$\geq$$t_{i,2}$ are the largest two singular values of the correlation tensors of $\rho_{i}(\forall i).$\\
	Trilocal inequality (Eq.(\ref{ineqs})) is violated if:
	\begin{equation}\label{up211}
		\mathcal{B}_{3-star}>2^{\frac{1}{3}}.
	\end{equation}
Violation of Eq.(\ref{up211})imply that the corresponding network correlations are fully network nonlocal.
\subsection{Mutual Information For Eve}
In a QKD protocol, let source $\mathbf{S}$ distribute a two-qubit state $\varrho_{AB}$ in between the trusted parties $A$ and $B.$ Let Eve(untrusted) have control over $\mathbf{S}.$ Let it prepare a purification($|\psi_P\rangle$,say) of \(\varrho\) which it can hold until after all public communication are made between $A,B$ and then make the optimal measurement. Via this purification strategy Eve can extract the maximum possible side information compatible with $\varrho_{AB}.$ This is because Eve can prepare many possible systems($\varrho^{'},$say) correlated with the systems of $A$ and $B.$ But on applying some quantum operation to its part of $|\psi_P\rangle\langle \psi_P|,$ Eve can obtain any extension($\varrho^{'}_{ABE}$) of \(\varrho_{AB}\). The purification strategy by Eve is thus often referred to as \textit{worst case eavesdropping strategy}\cite{dw1}. Hence if the QKD protocol can be secured under such an attack, it will be secure against any weaker adversarial strategy.
\par Let $A$ have obtained a raw-key bit. After its measurement, the joint classical-quantum state(between $A$ and Eve) is\cite{dw1}:
\begin{eqnarray}\label{wr1}
	\varrho_{AE}&=&\sum_{a}p(a)|a\rangle\langle a|\otimes\varrho_E^{(a)}
\end{eqnarray}
$\varrho_E^{(a)}$ is the conditional state of Eve when $A$ obtains output $a$ due to a projective measurement $M_A:$
\begin{eqnarray}\label{wr2}
	p(a)&=&\textmd{Tr}[M_A^{(a)}\otimes \mathbb{I}_B. \varrho_{AB}]\nonumber\\
	\varrho_E^{(a)}&=&\frac{\textmd{Tr}[M_A^{(a)}\otimes \mathbb{I}_B. |\psi_P\rangle\langle \psi_P|]}{p(a)},\,a=0, 1.
\end{eqnarray}
Eve obtains mutual information $I(A:E_e)$ where $E_e$ denotes the classical variable corresponding to any possible POVM($\{M_e\}$) that it performs.\\
Now, $\forall \{M_e\},$ $I(A:E_e)$ is upper bounded\cite{dw1} by Holevo quantity $\chi(A;E)$\cite{dw2}:
\begin{eqnarray}\label{wr3}
	I(A:E_e)&\leq& \chi(A:E)\nonumber\\
	\textmd{\small{where }} \chi(A:E)&=&S(\varrho_E)-\sum_{a}p(a)S(\varrho_E^{(a)})
\end{eqnarray} 
Here $S(.)$ denote Von-Neumann entropy and $\varrho_E$$=$$\textmd{Tr}_{AB}|\psi_P\rangle\langle \psi_P|.$\\
In a one-way reconciliation, w.l.o.g., let $A$ be the key-generating party and $B$  correct its data. The asymptotic $1$-way key rate is given by\cite{dw1}:
\begin{eqnarray}\label{wr4}
	R&=&H(A|E)-H(A|B)
\end{eqnarray} 
While $H(A|E)$ gives uncertainty of Eve about $A^{'}$s key, $H(A|B)$ is the information that $B$ needs to rectify errors.\\
If $R$$>$$0,$ then asymptotically $A$ and $B$ can extract secret key at that rate using $1$-way classical post-processing\cite{dw1}.
For a classical \(A\) and quantum Eve:
\begin{eqnarray}\label{wr5}
	H(A|E)&=&H(A)-\chi(A:E)
\end{eqnarray}
Again
\begin{eqnarray}\label{wr6}
	I(A:B)&=&H(A)-H(A|B).
\end{eqnarray}
So,
\begin{eqnarray}\label{wr7}
	R&=&H(A)-\chi(A:E)-(H(A)-I(A|B))\nonumber\\
	&=&I(A:B)-\chi(A:E)
\end{eqnarray}
As the purification is a pure state: $S(\varrho_{AE})$$=$$S(\varrho_{AB}).$ Again as $A$ performs projective measurements(rank one), so $S(\varrho_E^{(a)})$$=$$S(\varrho_B^{(a)}).$ Hence, Holevo quantity $\chi(A:E)$ can be computed without any knowledge of Eve's measurements:
\begin{eqnarray}\label{wr8}
	\chi(A:E)&=&S(\varrho_{AB})-\sum_{a}p(a)S(\varrho_B^{(a)}).
\end{eqnarray}
\subsection{No-Click Events}
Consider a bipartite Bell experimental set-up. Ideally, let each of the two parties $A,B$ are supposed to obtain two outputs $a,b$$\in$$\{0,1\}.$ However, in practical situations the detectors used by $A$ and $B$ sometimes fail to detect an incoming quantum system. Then the outcomes of each party can be any one from $\{0,1,\phi\}$ with $\phi$ considered as the labeling of \textit'{no-click event}'\cite{det1,det2}. Let $\eta$ denote the detection efficiency of the detector of each of $A$ and $B.$ Then:
\begin{equation}\label{detef1}
	Pr(\phi)=1-\eta;\,\,Pr('click')=\eta.
\end{equation}
The parties may either discard $\phi,$ under the assumption of fair sampling or may retain $\phi.$ Consider a strategy where the parties retain the \textit'{no-click event}' by assigning it a binary value:
\begin{equation}\label{detef2}
	\phi\rightarrow+1
\end{equation}
Denoting $A,B^{'}$s effective outcomes as $\tilde{a},\tilde{b},$ one have $\tilde{a},\tilde{b}$$\in$$\{0,1\}.$ For any input $x,y$ of $A,B,$ correlator term $E_{x,y}^{(id)}$now becomes\cite{det1,det2}:
\begin{eqnarray}\label{detef3}
	E_{x,y}(\eta)&=&\eta^2 E_{x,y}^{(id)}+(1-\eta)^2+\eta(1-\eta)(\langle A_x\rangle+\langle B_y\rangle)\nonumber\\
\end{eqnarray}
$E_{x,y}^{(id)}$ denote corresponding correlator in ideal set-up.\\
For the specific strategy of retaining $\phi,$ the outcomes remain dichotomic. So, the local bound of any Bell inequality remains unaltered. For example, Bell-CHSH inequality\cite{cla} is given by\cite{det1,det2}:
\begin{eqnarray}\label{detef4}
	E_{0,0}(\eta)+E_{0,1}(\eta)+E_{1,0}(\eta)-E_{1,1}(\eta)&\leq&2\\
	\quad\,\,E_{x,y}(\eta)\textmd{\small{ given by Eq.(\ref{detef3})}}&&\nonumber
\end{eqnarray}
\section{4-Party QKD Protocol}\label{trilov}
	The protocol to be designed here is based on a network involving four legitimate parties $A_1,A_2,A_3,A_4$ and three independent sources $\mathbf{S}_1,\mathbf{S}_2,\mathbf{S}_3.$ $A_1$ will be the \textit{central} party whereas others will be considered as \textit{extreme} parties. $\forall i,$ $\mathbf{S}_i$ prepares and distributes a two-qubit state between central party($A_1$) and extreme party $A_{i+1}.$ Independent qubit communication will take place across the network in the sense that preparation and distribution of qubits from $\mathbf{S}_i$ will be independent of that from $\mathbf{S}_j$$\forall i$$\neq$$j.$ This independence of qubits preparation and distribution corresponds to the \textit{trilocal constraint}.
	At the end of the protocol a block-structured secure key with block length $3$ will be shared between $A_1,A_2,A_3,A_4.$ However, such a key will not be a genuine multipartite key shared between all the four trusted parties.\\
	For rest of the paper, let $\mathcal{N}_4$ denote the QKD protocol.
	\subsection{Steps Of $\mathcal{N}_4$ }
	The steps of $\mathcal{N}_4$ are now detailed below:
	\begin{enumerate}
		\item \textit{Qubits Preparation and Distribution Stage:} $\mathbf{S}_i$ prepares $n$ identical copies of two-qubit entangled state $\rho_{i}(\forall i$$=$$1,2,3)$ and sends one qubit of $\rho_{i}$ to $A_1$ and another qubit to $A_{i+1}(i$$=$$1,2,3).$
		\item \textit{Measurement Stage:} 
		$A_1$ executes the following steps:
		\begin{enumerate}
\item[(i)]$\forall i$$=$$1,2,3,$ $A_1$ measures $h_1($$<$$n)$ copies of single qubit of $\rho_i,$ in one of two randomly chosen two-dimensional MUBs $B_{i}^{(1)},B_{i}^{(2)}.$ $\forall i=1,2,3,$ let $\mathcal{B}_{i}$$=$$\{B_{i}^{(k)}\}_{k=1}^2$ denote the collection of single-qubit MUBs used by $A_1.$
			\item [(ii)] For each of remaining $n-h_1$ copies, $A_1$ performs single projective measurement in tripartite GHZ basis on the joint state of the three qubits of $\rho_1,\rho_2,\rho_3.$ 
		\end{enumerate}
		Each of $A_2,A_3,A_4$ executes the following steps:
		\begin{enumerate}
			\item [(a)] $\forall i$$=$$1,2,3$, $A_{i+1}$ measures $h_{i+1}($$<$$n)$ copies of single qubit of $\rho_i,$ in one of two randomly chosen MUBs from $\mathcal{B}_{i}$$=$$\{B_{i}^{(k)}\}_{k=1}^2.$
			\item [(b)] $A_i$ performs single-qubit projective measurements randomly in any one of two arbitrary directions $\vec{m}_i.\vec{\sigma},\vec{n}_i.\vec{\sigma}$ on each of remaining $n$$-$$h_{i+1}$ copies of $\rho_i.$ 
		\end{enumerate}
		\item \textit{Measurement Setting Disclosure Stage:}
		The parties now broadcast the type of measurement each of them has performed. Precisely $A_1$ announces which of step 2(i) or 2(ii) it has performed in each of $n$ rounds. Similarly, $\forall i$$=$$1,2,3,$ $A_i$ announces which of steps 2(a) or 2(b) it has performed in each round. 
		They only keep the rounds in which $A_1$ performs 2(i) and each of $A_2,A_3,A_4$ performs 2(a) or $A_1$ performs 2(ii) and each of $A_2,A_3,A_4$ performs 2(b). Let $r_1,r_2$ denote the successful rounds corresponding to first type(steps 2(a) and 2(i) executed) and second type(steps 2(b) and 2(ii) executed)). So:
		\begin{eqnarray*}
			r_1&<&\textmd{Min}_{i=1}^4h_i\,\,\textmd{and } r_2<n-\textmd{Max}_{i=1}^4h_i.
		\end{eqnarray*}  
		\item \textit{Trilocal Inequality Testing Stage:} All the parties broadcast their outputs resulting from measurement steps 2(ii) and 2(b) in each of $r_2$ rounds. Using the $4$-partite measurement statistics $P(\bar{a}_1,a_2,a_3,a_4|x_2,x_3,x_4)$ arising in steps (2ii) and (2b), trilocal inequality(Eq.(\ref{ineqs})) is tested. Here, $\forall i=2,3,4$ $(x_i,a_i)$ denotes the input-output pair of $i^{th}$ trusted party $A_i$ and $\bar{a}_1$$=$$(a_{11},a_{12},a_{13})$ denotes $3$-bit output of fixed GHZ basis measurement by $A_1.$\\
		If violation of trilocal inequality(Eq.(\ref{ineqs})) is observed then the next step of the protocol is executed. Otherwise, the protocol is aborted.
		\item \textit{Sifting Stage:} All the trusted parties now publicly announce their chosen MUBs for each of $r_1$ rounds in steps(2i) and (2a) so as to compare whether the MUB used by $A_1$ on qubit of $\rho_i$ is identical with the MUB used by $A_{i+1}$ on its qubit of $\rho_i$ $\forall i$$=$$1,2,3.$ In particular:
		\begin{itemize}
			\item \textit{Bases Reconciliation:} $\forall i$$=$$1,2,3,$ single-qubit MUB used by $A_1$ for measuring qubit corresponding to state $\rho_i$ is compared with MUB used by $A_{i+1}.$ \\
			Out of $r_1$ turns, the parties discard their measurement outputs  where for at least one $i$$\in$$\{1,2,3\},$ MUBs chosen by $A_1$ is not same as that chosen by $A_{i+1}$ from $\mathcal{B}_{i}.$\\
			\item \textit{Sifting Keys Generation:} Let $r_3$$<$$r_1$ denote the turns when MUB(chosen from $\mathcal{B}_{i}$) used by $A_1$ and $A_{i+1}$ over $\rho_i$ are identical $\forall i=1,2,3.$ \\
			In any of these $r_3$ turns, let $B_i^{(j_i)}$($j_i$$\in$$\{1,2\}$) denote the basis chosen by $A_1$ and $A_{i+1}$ to measure $\rho_i$($i$$=$$1,2,3$). Let $|\phi_{i,k_{i}}^{(j_i)}\rangle,$ $|\phi_{i,s_{i}}^{(j_i)}\rangle$ denote output obtained by $A_1$ and $A_{i+1}$ respectively.
			Each basis($B_i^{(j_i)}$) being two dimensional, $k_{i},s_{i}$$\in$$\{1,2\}.$\\
			Binary labelings are used to denote these outputs: $\forall i$$=$$1,2,3$ and $k$$=$$1,2,$ let $|\phi_{i,1}^{(k)}\rangle\rightarrow 0$ and $|\phi_{i,2}^{(k)}\rangle\rightarrow 1$ Denoting these outputs as bits, each of the four parties now has a bit string:
			\begin{itemize}
				\item Each of $A_2,A_3,A_4$ has a bit string of length $r_3.$
				\item $A_1$ has a block-structured bit string of length $r_3$ with block length $3$. So the total length of $A_1^{'}$s bit string is $3r_3.$
			\end{itemize} 
			These bit strings of the trusted parties are referred to as \textit{sifted keys} generated in the protocol.
			\item \textit{QBER Generation:} $\rho_1,\rho_2,\rho_3$ are two-qubit entangled states. So, ideally, in each of $r_3$ turns in the sifting stage, $\forall i=1,2,3,$ outputs obtained from $\rho_i$ due to $A_1$ and $A_{i+1}^{'}$s local measurements in identical MUBs, are supposed to be identical.
			Hence, in terms of binary labelings, $\forall m$$=$$1,2,\cdots,r_3$ and $\forall i=1,2,3,$ bit value at $i^{th}$ position in the $m^{th}$ block of $A_1^{'}$s sifted key is supposed to be same as the bit value in $m^{th}$ position of $A_{i+1}^{'}$s sifted key. \\
			However, in practical scenarios, due to channel noise, imperfect devices and presence of eavesdropper, the outputs of $A_1$ differ from that of the other parties for each of $r_3$ turns. This in turn leads to generation of QBER(Eq.(\ref{basis2})) in the protocol. The parties use a portion of their sifted keys to compute QBER. If QBER exceeds critical value($Q_0$) of QBER(to be discussed later) for the protocol, then the protocol is aborted. Otherwise, next step of the protocol is executed.
		\end{itemize}
		\item \textit{Generation of Secret Key:} Remaining part of the sifted keys of the trusted parties are then subjected to error correction and privacy amplification procedures so as to extract a shorter but secure
		secret key shared among the trusted parties. The secured key is block-structured bit string of length $r_4$(say with $r_4$$<$$r_3$) with block length $3$. So the total length of the secret key bit string is $3r_4.$
	\end{enumerate}
	\subsubsection{Secret Key Not Genuinely $4$-partite}
		 To this end it should be noted that the secret key(length $3r_4$) generated in the protocol can be interpreted as a concatenation of $3$ individual bipartite secret keys, each of length $r_4$ and shared between $A_1$ and a extreme party. So conceptually here the protocol does not establish a truely $4$-partite secret key which is usually created in a multipartite key distribution protocol(Quantum Conference Key Agreement Protocol, for instance). 

	\subsection{Expression of QBER In $\mathcal{N}_4$}
	In sifting stage, $\forall i$$=$$1,2,3,$ central party $A_1$ chooses from the collection($\mathcal{B}_{i}$) of two MUBs $B_i^{(1)},B_i^{(2)}$ to measure qubit of $\rho_i$ Similarly, each of the extreme parties $A_{i+1}$ chooses from $\mathcal{B}_{i}$ to measure qubit of $\rho_1,\rho_2,\rho_3$ respectively. 	
	$\rho_1,\rho_2,\rho_3$ being two-qubit entangled states in each of $r_3$ turns, ideally in sifting stage, $\forall i=1,2,3,$ outputs obtained from $\rho_i$ due to $A_1$ and $A_{i+1}^{'}$s local measurements in identical MUBs, are supposed to be same. So in each of $r_3$ turns, on measuring $\rho_i$($\forall i$$=$$1,2,3$) in MUB $B_i^{(j_i)},$($\forall j_i$$=$$1,2$), $A_1$ and $A_{i+1}$ are both supposed to obtain same output $|\phi_{i,k_{i}}^{(j_i)}\rangle$ with $k_{i}$$\in$$\{1,2\}.$ Precisely, $\forall j_i$$\in$$\{1,2\},$ when $A_1$ obtains $|\phi_{1,k_{i}}^{(j_i)}\rangle,$ $|\phi_{2,k_{i}}^{(j_i)}\rangle$ and $|\phi_{3,k_{i}}^{(j_i)}\rangle$ after measuring $\rho_1,\rho_2,\rho_3$ respectively, parties $A_2,A_3,A_4$ are also supposed to obtain the same outputs: $|\phi_{1,k_{i}}^{(j_i)}\rangle,$ $|\phi_{2,k_{i}}^{(j_i)}\rangle$ and $|\phi_{3,k_{i}}^{(j_i)}\rangle$ respectively.\\
	Let $\mathcal{O}_{1,i}^{(j_i)}$ denote the operator corresponding to measurement in basis $B_{i}^{(j_i)}$:
	\begin{eqnarray}\label{cr1}
		\mathcal{O}_{i}^{(j_i)}&=&\vec{U}_{i}^{(j_i)}\cdot \vec{\sigma}\,\,j_i=1,2\,\forall i=1,2,3.\nonumber\\
	\end{eqnarray}
	In any turn, when for at least one $i$$\in$$\{1,2,3\},$ there is mismatch in output obtained by $A_1$ with that of output of $A_{i+1},$ both measuring in MUB $B_i^{(j_i)}$ on their respective qubits of state $\rho_i,$ QBER($\mathbf{Q}$) is generated:
	\begin{eqnarray}\label{cr2}
		\mathbf{Q}&=&\frac{1}{2^3}\sum_{\substack{j_i=1,2\\\forall i=1,2,3}}(\sum_{S}\Pi_{i=1}^3\langle \phi_{i,k_i}^{(j_i)}\phi_{i,s_{i}}^{(j_i)}|\rho_i| \phi_{i,k_i}^{(j_i)}\phi_{i,s_{i}}^{(j_i)}\rangle\\
		&&\textmd{\small{where}}\nonumber\\
		S&=&\{k_1,k_2,k_3,s_1,s_2,s_3=1,2|k_i\neq s_i\,\textmd{\small{for at least one }}i\}\nonumber
	\end{eqnarray}
	Clearly, $\mathbf{Q}$ turns out to be the product of mismatch of outputs obtained from $\rho_1,\rho_2$ and $\rho_3,$ averaged over all possible combinations of MUBs chosen by the parties. So computation of QBER involves two summations: one over all possible MUBs from $\mathcal{B}_1,\mathcal{B}_2,\mathcal{B}_3$ while the other over all non-identical outputs from at least one of $\rho_1,\rho_2,\rho_3.$ \\
	For any $i$$\in$$\{1,2,3\},$ and $j_i$$\in$$\{1,2\},$ 
		\begin{eqnarray}\label{crp2}
	\sum_{\substack{k_i,s_i=1,2\\\forall i=1,2,3}}\Pi_{i=1}^3\langle \phi_{i,k_i}^{(j_i)}\phi_{i,s_{i}}^{(j_i)}|\rho_i| \phi_{i,k_i}^{(j_i)}\phi_{i,s_{i}}^{(j_i)}\rangle&=&1
	\end{eqnarray}
	Using Eqs.(\ref{cr2},\ref{crp2}), QBER can be expressed alternatively:
	\begin{eqnarray}\label{cr3}
		\mathbf{Q}&=&\frac{1}{2^3}\sum_{\substack{j_i=1,2\\\forall i=1,2,3}}(1-\sum_{\substack{k_i=1,2\\\forall i=1,2,3}}\Pi_{i=1}^3\langle \phi_{i,k_i}^{(j_i)}\phi_{i,k_i}^{(j_i)}|\rho_i|\phi_{i,k_i}^{(j_i)} \phi_{i,k_{i}}^{(j_i)}\rangle)\nonumber\\
		&=&1-\frac{1}{2^3}\sum_{\substack{j_i,k_i=1,2\\\forall i=1,2,3}}\Pi_{i=1}^3\langle \phi_{i,k_i}^{(j_i)}\phi_{i,k_i}^{(j_i)}|\rho_i|\phi_{i,k_i}^{(j_i)} \phi_{i,k_{i}}^{(j_i)}\rangle
	\end{eqnarray}
Using the parameters of single-qubit measurement $\mathcal{O}_{i}^{(j_i)}$(Eq.(\ref{cr1})), QBER(Eq.(\ref{cr3})) takes the form:
	\begin{eqnarray}\label{cr4}
		\mathbf{Q}&=&1-\frac{1}{2^3}\sum_{\substack{j_i,k_i=1,2\\\forall i=1,2,3}}(\Pi_{i=1}^3(\frac{\mathbb{I}_2+(-1)^{k_i+1}\vec{U}_{i}^{(j_i)}\cdot \vec{\sigma}}{2^2})^{\otimes 2}\cdot \rho_i)\nonumber\\
		&=&1-\frac{1}{2^6}\sum_{\substack{j_i=1,2\\\forall i=1,2,3}}\Pi_{i=1}^3(1+\vec{U}_{i}^{(j_i)}\cdot T_i\vec{U}_{i}^{(j_i)})\,\textmd{\small{By Eq.(\ref{st4})}}
	\end{eqnarray}
	The parties estimate $\mathbf{Q}$ after comparing a small portion of their sifting key.
	\subsection{Minimization of QBER}
	In any QKD protocol, the QBER($\mathbf{Q}$) must be kept as low as possible. This is because an increased $\mathbf{Q}$ directly reflects loss of correlations between the trusted parties, which can arise from both channel noise and eavesdropping. As any eavesdropping attempt inevitably introduces errors, minimizing and finding threshold QBER is fundamentally linked to the security of the protocol.
	Clearly, minimizing $\mathbf{Q}$ in Eq.(\ref{cr4}) is equivalent to maximizing $\mathbf{H}$ given by:
	\begin{eqnarray}\label{cr5}
		\mathbf{H}&=&\sum_{\substack{j_i=1,2\\\forall i=1,2,3}}\Pi_{i=1}^3(1+\vec{U}_{i}^{(j_i)}\cdot T_i\vec{U}_{i}^{(j_i)})
	\end{eqnarray}
	It is clear from the expression of $\mathbf{H}$(Eq.(\ref{cr5})) that to maximize it, one needs to perform two levels of maximization: over all measurement directions $\vec{U}_i^{(j_i)}$ corresponding to two-dimensional MUBs and also over all state parameters(particularly correlation tensors $T_i$) corresponding to two-qubit states $\rho_i$($\forall i$). 
	\subsubsection{Maximizing $\mathbf{H}$ over MUBs} As discussed before, in first step of measurement stage in the protocol, $\forall i$$=$$1,2,3,$ $A_1$(in step 2(i)) and $A_{i+1}$(in step 2(a)) chooses randomly from a collection of two single-qubit MUBs($\mathcal{B}_i$). Now, in case of local dimension $d$$=$$2,$ up to global phase factor, there exist only $3$ possible MUBs\cite{mub1}:$\{\frac{|0\rangle\pm|1\rangle}{2}\},$  $\{\frac{|0\rangle\pm \imath\,|1\rangle}{2}\}$ and $\{|0\rangle,|1\rangle\}.$ Hence, maximization over $\vec{U}_i^{(1)},\vec{U}_i^{(2)}(i$$=$$1,2,3)$ gives:
	\begin{eqnarray}\label{cr6}
		\mathbf{H}&\leq& \sum_{\substack{j_i=1,2\\\forall i=1,2,3}}\Pi_{i=1}^3(1+t_{i,j_i}).
	\end{eqnarray}
	In Eq.(\ref{cr6}), $t_{i,1},t_{i,2}$ stand for the largest two ordered singular values($t_{i,1}$$\geq$$t_{i,2}$$\geq$$t_{i,3}$) of correlation tensor $T_i$ of $\rho_i(\forall i).$
	\subsubsection{Maximizing $\mathbf{H}$ over State Parameters}
	It is clear from above bound(Eq.(\ref{cr6})) that $\mathbf{H}$ next needs to be maximized with respect to $t_{1,1},t_{1,2},$$t_{2,1},t_{2,2}$$t_{3,1}$ and $t_{3,2}.$ Correlations obtained due to second step(2(ii) and 2(b)) in measurement stage are used to check violation of trilocal inequality(Eq.(\ref{ineqs})) in $4^{th}$ step of the protocol. It becomes important to find out the minimum possible QBER in case no violation of trilocal inequality(Eq.(\ref{ineqs})) is observed in the protocol. This is because once the minimum possible QBER is known in absence of trilocal inequality's violation, then QBER can be reduced further if violation is observed. Let $\mathbf{Q}_0$ denote the minimum value of $\mathbf{Q}$ and let it be referred to as \textit{critical value} of QBER in $\mathcal{N}_4.$ So, in absence of detectable non-trilocality among $A_1,A_2,A_3,A_4$, $\mathbf{Q}$ can never be made less than $\mathbf{Q}_0.$ However, the error rate($\mathbf{Q}$) can be reduced below $\mathbf{Q}_0$ once the trusted parties in the network($\mathcal{N}_4$) detect non-trilocal correlations.
	\paragraph{Identical States:} Let $\rho_1,\rho_2,\rho_3$ used in $\mathcal{N}_4$ be identical: $\rho_1$$=$$\rho_2$$=$$\rho_3$$=\varrho$(say). Here, Eq.(\ref{cr6}) gets simplified:
	\begin{eqnarray}\label{cr6i}
		\mathbf{H}&\leq&	\sum_{i=1}^2(1+t_i)^3+3(1+t_1)^2(1+t_2)\nonumber\\
		&&+3(1+t_2)^2(1+t_1)\nonumber\\
		&=&(2+t_1+t_2)^3,
	\end{eqnarray}
	where $t_1,t_2$ denote the two largest singular values of correlation tensor of $\varrho.$\\
	The criterion corresponding to no violation of Eq.(\ref{ineqs}) puts restriction over $t_{i,k}.$ Maximization of $\mathbf{H}$ over $t_{i,j}$ thus becomes a constrained maximization problem.\\
	Let $\mathbf{H}_0$ denote the maximum value of $\mathbf{H}$ resulting from such optimization. Using $\mathbf{H}_0$ in Eq.(\ref{cr5}), one gets the minimum value of QBER possible in the protocol. Theorem below provides $\mathbf{Q}_0$ for $\mathcal{N}_4$.
	\begin{theorem}\label{theo1}
		In the network based $4$-party QKD protocol $\mathcal{N}_4$, involving identical states, QBER generated cannot be less than $\mathbf{Q}_0$$=$$1-\frac{\sqrt{2}(1+2^{\frac{1}{6}})^3}{16}\approxeq$$0.154887$ when the trusted parties do not observe violation of trilocal inequality(Eq.(\ref{ineqs})).
	\end{theorem}
	\textbf{Proof:}See Appendix.A\\
	Above theorem provides $\mathbf{Q}_0$ when all the states shared among the trusted parties are identical. However, considering distribution of qubits from non-identical states is also important from practical view point. 
Next theorem provides critical value of $\mathbf{Q}$ for this type of scenarios where $\mathcal{N}_4$ involves non identical states. 
	\begin{theorem}\label{theo11}
		If all three states used in $\mathcal{N}_4$ are not identical then QBER generated cannot be less than $\mathbf{Q}_0$$=$$1-\frac{3+(2^{\frac{2}{3}}-1)^{\frac{3}{2}}}{4}$$ \approxeq$$ 0.13745$ when the trusted parties do not observe violation of trilocal inequality(Eq.(\ref{ineqs})).
	\end{theorem}
	\textbf{Proof:}See Appendix.B.\\
	Comparison of the threshold values provided by the two theorems points out that when all the states are not identical then $\mathbf{Q}$ can be reduced more($13.7\%$ approx) than that possible when all the states are identical($15.5\%$ approx). 
	$\mathbf{Q}_0$ will next be used to frame a criterion to test validity of the protocol in the sifting stage. 
	\subsection{Criteria Necessary For Validating $\mathcal{N}_4$}
As per design of the protocol($\mathcal{N}_4$), following two security checks are executed:
	\begin{itemize}
		\item \textit{First Check:} In $4^{th}$ step using trilocal inequality(Eq.(\ref{ineqs})).
		\item \textit{Second Check:} In $5^{th}$ step using $\mathbf{Q}_0$ provided by Theorems.\ref{theo1},\ref{theo11}.
	\end{itemize}
Above check steps themselves do not establish the security of the QKD protocol as both of these criteria are necessary but not sufficient. However, these can be used as validation metrics to assess whether the protocol satisfies the necessary conditions required for secure key generation.  So, these checks will be further referred to as \textit{ protocol validation checks}. A complete security proof of the protocol lies beyond te scope of the present work. 
	\subsubsection{First Validation Check}
This check relies upon exploiting the fact that violation of Eq.(\ref{ineqs}) ensures existence of genuine form of network non trilocality among $A_1,A_2,A_3,A_4.$ Such correlations thus cannot be obtained even if only one party does not share any non-classical correlation with the other three\cite{bilo6}. If each of $\rho_1,\rho_2,\rho_3$ be pure entangled state, then Eq.(\ref{ineqs}) is supposed to be violated\cite{bilo6} if the four partite correlations shared among the trusted parties are untampered. In case of an eavesdropper's interference, these correlations are disturbed and the inequality may no longer be violated. $\mathcal{N}_4$ is thus aborted if the parties do not observe any such violation.\\
	Now, as discussed in sec.\ref{pre}, Eq.(\ref{ineqs}) is violated if $\rho_1,\rho_2,\rho_3$ satisfy:
	\begin{eqnarray}\label{cr7}
		\Pi_{i=1}^3 (t_{i,1})^\frac{2}{3}+\Pi_{i=1}^3 (t_{i,2})^\frac{2}{3}>2^{\frac{2}{3}}
	\end{eqnarray}
	Eq.(\ref{cr7}) thus acts as a criterion($\mathcal{C}_{\mathcal{N},1},$say) to check validity of $\mathcal{N}_4.$\\
	$\mathcal{C}_{\mathcal{N},1}$:\textit{if Eq.(\ref{cr7}) is satisfied in $4^{th}$ step, then next step of the protocol is executed. Otherwise it is aborted.} 
	\subsubsection{Second Validation Check}
	QBER($\mathbf{Q}$) quantifies the disturbance introduced in any QKD protocol. As stated in above theorems, as long as there is no violation of trilocal inequality in $\mathcal{N}_4$, QBER can never be less than $\mathbf{Q}_0.$ Now sifting step in $\mathcal{N}_4$ is executed only after violation of trilocal inequality is observed in $4^{th}$ step. When such violation is observed, QBER can be less than the critical value($\mathbf{Q}_0$). So, the protocol must be aborted in the sifting step whenever $\mathbf{Q}$ exceeds this threshold. Precisely second necessary criterion is given by:
	\begin{eqnarray}\label{dummy}
		\mathbf{Q}&<&\mathbf{Q}_0.
	\end{eqnarray}
	If above criterion(Eq.(\ref{dummy})) is not satisfied then the correlations between the trusted parties may become too weak to guarantee secrecy in the protocol. \\
	In sifting step of $\mathcal{N}_4,$ minimizing over all possible collection($\mathcal{B}_1,$$\mathcal{B}_2,$$\mathcal{B}_3$) of two MUBs, $\mathbf{Q}$ is given by Eq.(\ref{cr6}):
	\begin{eqnarray}\label{cr8}
		\mathbf{Q}^{'}=\mathbf{Q}]_{\substack{\textmd{\tiny{Min }}\\\textmd{\tiny{over MUBs}}}}&=&1-\frac{1}{2^6}\sum_{\substack{j_i=1,2\\\forall i=1,2,3}}\Pi_{i=1}^3(1+t_{i,j_i}).
	\end{eqnarray}
	$\mathbf{Q}_0,$ provided by Theorem.\ref{theo1}, results after second level of minimization of $\mathbf{Q}.$ Hence, $\mathbf{Q}_0$ is obtained by  minimizing $\mathbf{Q}^{'}$ over all state parameters assuming no violation of trilocal inequality. Now, when the trilocal inequality is violated then $\mathbf{Q}^{'}$ can be lesser than $\mathbf{Q}_0$ and this is considered as second validation criterion(Eq.(\ref{dummy})). \\
	In general, the states($\rho_i$) shared between each of the three pairs of central and extreme parties($(A_1,A_{i+1})$) are supposed to be non-identical. Critical error rate provided by Theorem.\ref{theo11} is to be used to frame the validation criterion. Using $\mathbf{Q}_0$(Theorem.\ref{theo11}) and $\mathbf{Q}^{'}$(Eq.(\ref{cr8})) in Eq.(\ref{dummy}), one gets:
	\begin{eqnarray}\label{cr9}
		1-\frac{1}{2^6}\sum_{\substack{j_i=1,2\\\forall i=1,2,3}}\Pi_{i=1}^3(1+t_{i,j_i})&<&1-\frac{16(3+(2^{\frac{2}{3}}-1)^{\frac{3}{2}})}{2^6}\nonumber\\
		\Rightarrow	\sum_{\substack{j_i=1,2\\\forall i=1,2,3}}\Pi_{i=1}^3(1+t_{i,j_i})&>& 16(3+(2^{\frac{2}{3}}-1)^{\frac{3}{2}}).
	\end{eqnarray}
	Above relation explicitly gives the second necessary criterion($\mathcal{C}_{\mathcal{N},2},$ say) for protocol validation.\\
	$\mathcal{C}_{\mathcal{N},2}$:\textit{if Eq.(\ref{cr9}) is satisfied in $5^{th}$ step, $\mathcal{N}_4$ is then used to generate secure key. Otherwise it is aborted.}\\
	In case, the trusted parties do share identical states they use $\mathbf{Q}_0$ provided by Theorem.\ref{theo1} instead of Theorem.\ref{theo11}. Using $\mathbf{H}$ provided by Eq.(\ref{cr6i}), the expression of $\mathbf{Q}^{'}$(Eq.(\ref{cr8})) in this case gets simplified:
	\begin{eqnarray}\label{cr8i}
		\mathbf{Q}^{'}=\mathbf{Q}]_{\substack{\textmd{\tiny{Min }}\\\textmd{\tiny{over MUBs}}}}&=&1-\frac{1}{2^6}(2+t_1+t_2)^3.
	\end{eqnarray}
	For above expression of $\mathbf{Q}^{'}$(Eq.(\ref{cr8i})) and critical error rate provided by Theorem.\ref{theo1}, one gets:
	\begin{eqnarray}\label{cr9i}
		1-\frac{1}{2^6}(2+t_1+t_2)^3&>&1-\frac{4\sqrt{2}(1+2^{\frac{1}{6}})}{2^6}\nonumber\\
		\Rightarrow t_1+t_2&>&2^{\frac{5}{6}}.
	\end{eqnarray}
	Here, the second necessary criterion($\mathcal{C}_{\mathcal{N},2}^{'},$say) thus takes the same form as $\mathcal{C}_{\mathcal{N},2}$ with only Eq.(\ref{cr9}) now replaced by Eq.(\ref{cr9i}).\\
	When the trusted parties have prior information that they are supposed to share identical states($\varrho$) then they may use $\mathcal{C}_{\mathcal{N},2}^{'}$ to check whether $\varrho$ can be used to execute the protocol successfully or not. For instance, let $A_1$ generate and distribute qubits of identical states $\rho_1$$=$$\rho_2$$=$$\rho_3$$=$$\varrho$ among $A_2,A_3,A_4.$ Let the parties know that the qubits are communicated($A_1$ to $A_{i+1}$) through identical noisy channel($\Lambda$,say). Let $\rho_i^{'}$$=$$\Lambda(\rho_i).$ Here, in absence of any untrusted party, state shared among $A_1$ and $A_{i+1}(\forall i$$=$$1,2,3)$ is supposed to be the same $\rho_1^{'}$$=$$\rho_2^{'}$$=$$\rho_3^{'}$$=$$\varrho^{'}.$ To check whether $\varrho^{'}$ can be used to execute $\mathcal{N}_4$ the parties will use $\mathcal{C}_{\mathcal{N},2}^{'}.$ However, if the parties use $\mathcal{C}_{\mathcal{N},2}$ instead of $\mathcal{C}_{\mathcal{N},2}^{'}$ then $\varrho^{'}$ may fail to succeed second validation check. Consequently the parties abort the protocol even though no malicious party is present. Choice of second checking criterion thus plays an important role for segregating utility of two qubit states in $\mathcal{N}_4.$ This idea will be further illustrated during characterization of two-qubit states in next section. 
	\par Now both the constraints $\mathcal{C}_{\mathcal{N},1}$ and $\mathcal{C}_{\mathcal{N},2}$(or $\mathcal{C}_{\mathcal{N},2}^{'}$) are necessary as both of these have to be satisfied in order to use $\mathcal{N}_4.$ But these are not sufficient to ensure unconditional security as even if both of these hold, Eve can still interfere without getting detected. $\mathcal{C}_{\mathcal{N},1}$ and $\mathcal{C}_{\mathcal{N},2}$(or $\mathcal{C}_{\mathcal{N},2}^{'}$) will be referred to as \textit{validation criteria.} 
They act as practical first-level security witnesses that $\mathcal{N}_4$ must satisfy before applying more refined security strategies. Moreover these criteria will aid in analyzing features of quantum states that can be used for successful execution of the QKD protocol. Present work focuses in the latter direction. 

	\section{Characterizing Quantum States Used In $\mathcal{N}_4$}\label{states}
	Ideally, $\rho_1,\rho_2,\rho_3$ used in $\mathcal{N}_4$ are supposed to be pure entangled states. However, owing to interaction with noisy environment, mixed entangled states get distributed among the parties even in absence of any malicious party. In context of practical situation it thus becomes pertinent to characterize the two-qubit states that can be used in $\mathcal{N}_4$ to generate secure key. \\
	The restrictions provided by Eqs.(\ref{cr7},\ref{cr9}) both involve state parameters only. Consequently, these criteria act as tool to detect utility of arbitrary two qubit states for designing $\mathcal{N}_4.$ Precisely, let $\rho_1,\rho_2,\rho_3$ be three arbitrary two-qubit states shared among the trusted parties $(A_1,A_2),(A_1,A_3),(A_1,A_4)$ respectively in $\mathcal{N}_4.$ If the correlation tensors($T_i)$ of $\rho_i$ satisfy both Eq.(\ref{cr7}) and Eq.(\ref{cr9}) then the protocol($\mathcal{N}_4$) is executed successfully. \\
	\subsection{Identical $\rho_1,\rho_2,\rho_3$}
	Let $\rho_i$$=$$\varrho$ be the state shared among $A_1$ and $A_{i+1}(\forall i).$ Let $T$$=$$\textmd{diag}(t_1,t_2,t_3)$ be the correlation tensor of $\varrho$ with $t_1$$\geq$$t_2$$\geq$$t_3$ denoting ordered singular values of $T$. In terms of $T,$ first  criterion takes the following form:
	\begin{eqnarray}\label{cr10}
		t_1^2+t_2^2>& 2^{\frac{2}{3}}\,\,\,\textbf{\textmd{\small{First criterion}}}
	\end{eqnarray}
	and second checking criterion is given by Eq.(\ref{cr9i}).\\
	When Eq.(\ref{cr10}) is violated then corresponding state $\varrho$ fails to generate detectable full network non-trilocal correlations\cite{bilo6}. For remaining discussion any two-qubit state($\varrho$) that violates Eq.(\ref{cr10}) will be referred to as \textit{trilocal state.} Otherwise it will be a \textit{fully network non-trilocal state} or simply \textit{fully network nonlocal(FNN) state.}  \\ 
	It is evident from previous discussion that analyzing the two-dimensional space corresponding to the largest two singular values $(t_1,t_2)$ of $T$ suffices for the characterization. For simplicity, further discussion with $t_1,t_2$ does not consider the ordering $t_1$$>$$t_2$ and thereby maintain a symmetry between them.   \\
	$\varrho$ is an arbitrary two-qubit state. Hence, $t_{1},t_{2}$$\in$$[0,1]$ define a unit square($\mathbf{S}$):
	\begin{equation}\label{basis8i}
		\mathbf{S}=\{(t_{1},t_{2}):0\leq t_{1},t_{2}\leq1\}
	\end{equation}
	Eq.(\ref{cr10}) represents a circle($C$,say) with center at the origin. Eq.(\ref{cr9i}) denotes a tangent line($\mathbf{T},$say) to $\mathbf{C}$(see Fig.\ref{fig2}) at point $(2^{-\frac{1}{6}},2^{-\frac{1}{6}}).$ \\
	Any density matrix corresponding to a point lying outside $\mathbf{S} $(Eq.(\ref{basis8i})) does not represent any two-qubit state(see Fig.\ref{fig2}). 
	For rest of the analysis, let $\varrho_{P}$ denote a two-qubit state corresponding to any point $P$ inside $\mathbf{S}.$ A portion of positive quadrant($\mathbf{C}_+$) only of $\mathbf{C}$ lies inside $\mathbf{S}.$ Trilocal states reside on and inside $\mathbf{C}_{+}.$ So any point $P$ lying inside $\mathbf{S}$ but outside $\mathbf{C}_{+}$ corresponds to a fully network nonlocal(FNN) state(see Fig.\ref{fig2}). Consequently when such $\varrho_P$ is used in $\mathcal{N}_4$ then the protocol passes in first validation check. However, the protocol may still fail in second validation check in case $P$ lies below the tangent line $\mathbf{T}$ as estimated $\mathbf{Q}$ exceeds or is at most equal to the critical error rate($\mathbf{Q}_0$). In case any state $\varrho_P$ can be used to successfully execute $\mathcal{N}_4,$ corresponding point $P$ in $\mathbf{S}$ must lie above $\mathbf{T}.$
	\begin{center}
	\begin{figure}
		\includegraphics[width=3in]{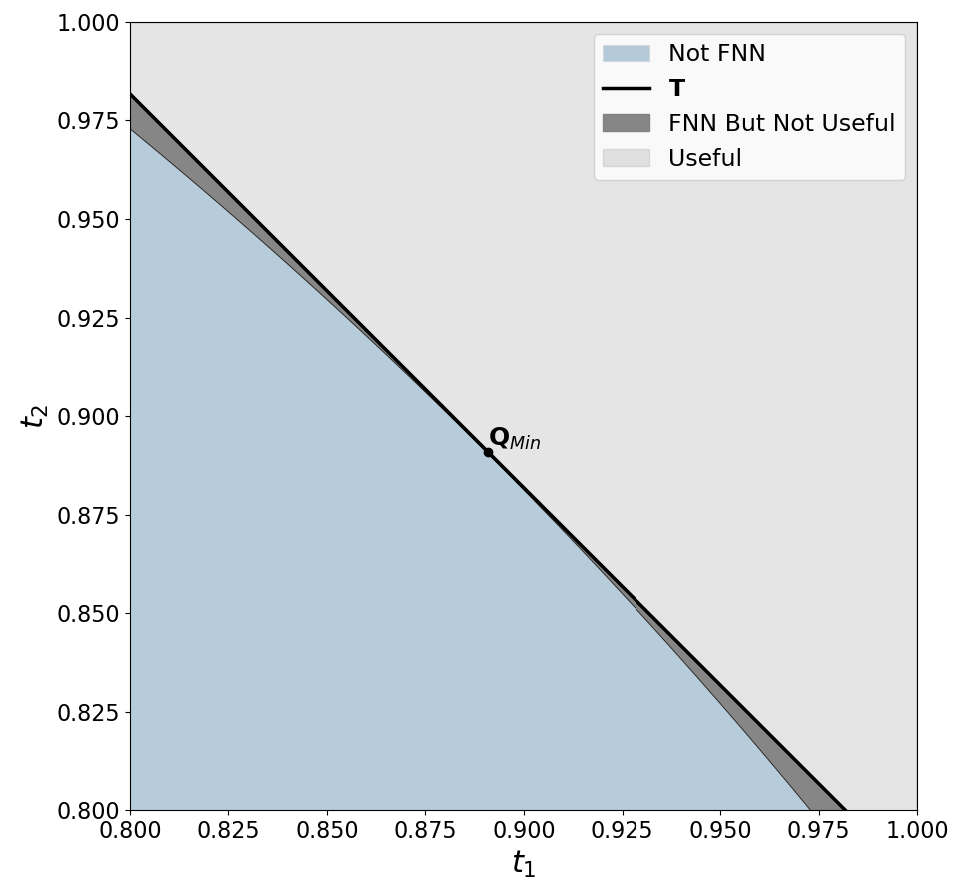}\\
		\caption{\emph{The possible state space corresponding to the largest two singular values($t_1,t_2$) of correlation tensor($T$) of an arbitrary two-qubit state is categorized into different sectors based on  validation criteria Eqs.(\ref{cr9i},\ref{cr10}). 
				Only for any point $P$ lying above $\mathbf{T}$ in $\mathbf{S},$ $\varrho_P$ can be used to execute $\mathcal{N}_4$ successfully. Point $\mathbf{Q}_{Min}$ represents the minima($2^{-\frac{1}{6}}$,$2^{-\frac{1}{6}}$) of QBER($\mathbf{Q}$) under assumption of no violation.}} 
		\label{fig2}
	\end{figure}
\end{center}
	\subsection{Illustration With Non-identical $\rho_1,\rho_2,\rho_3$}
	Next consider an instance of non-identical states ($\rho_1,\rho_2,\rho_3$) in $\mathcal{N}_4$:
	\begin{eqnarray}\label{ext1}
		T_1&T_2= =&\textmd{diag}(0.95,t_{j,2},t_{j,3});\,\,
		T_3=\textmd{diag}(0.96,t_{3,2},t_{3,3})
	\end{eqnarray}	
	Corresponding $\rho_1,\rho_2,\rho_3$ are useful in $\mathcal{N}_4$ if $t_{1,2},t_{2,2},t_{3,2}$ satisfy both the criteria(Eqs.(\ref{cr7},\ref{cr9})):
	\begin{eqnarray}\label{cr12}
		t_ {1, 2} t_ {2, 2} t_ {3, 2}&>&0.558986
	\end{eqnarray}
	\begin{eqnarray}\label{cr13}
		8.7025 t_{3,2} + t_{2,2}(8.732 + 2.95 t_{3,2}) + &&\nonumber\\
		t_{1,2} (8.732 + 2.95 t_{3,2} + t_{2,2} (2.96 + t_{3,2}))	&>&29.4437.
	\end{eqnarray}
	Clearly, there exist state parameters such that the states can be used to successfully run the protocol(see Fig.\ref{fig3}). 
	For a particular instance, let $T_1$$ =$$\textmd{diag}(0.95,0.91,0.9),$ $T_2$$ =$$\textmd{diag}(0.95,0.88,0.85)$ and $T_3$$ =$$\textmd{diag}(0.96,0.85,0.82).$ L.H.S. of Eq.(\ref{cr12}) and Eq.(\ref{cr13}) turn out to be $0.68068$ and $30.5669.$ respectively. Both the validation criteria are satisfied. Hence, these specific states can be used to design $\mathcal{N}_4.$
	\begin{center}
		\begin{figure}
			\includegraphics[width=3.1in]{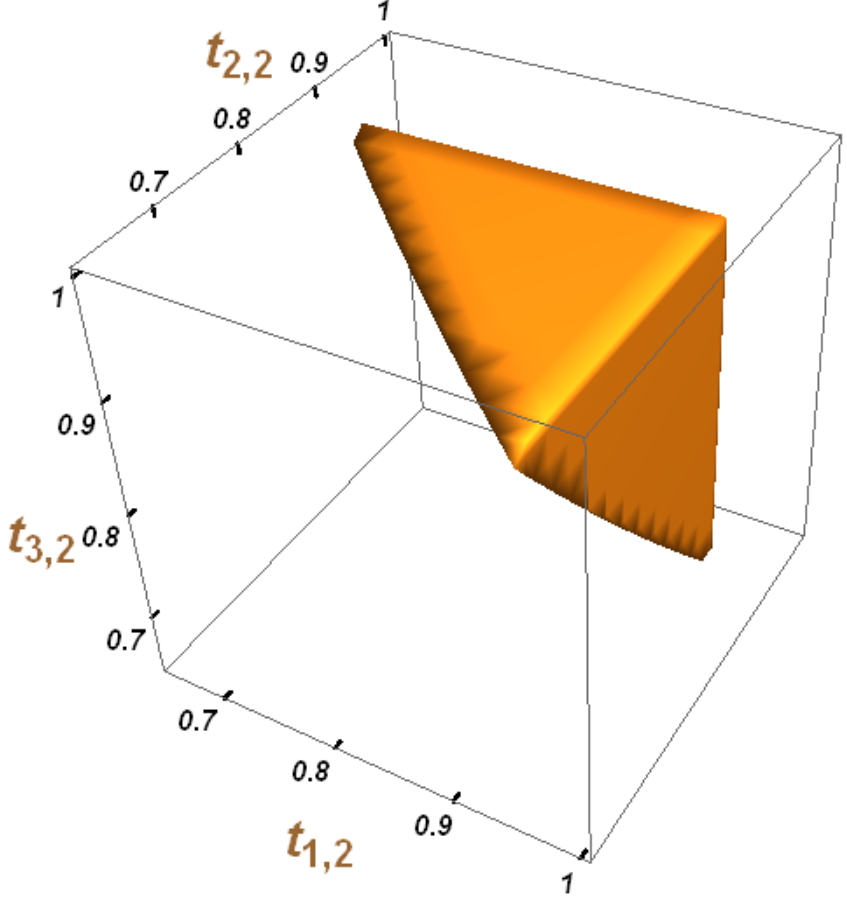} 
			\caption{\emph{Shaded region forms a part of three-dimensional space formed by $2^{nd}$ largest singular value $t_{i,2}$ of correlation tensor $T_i$ of $\rho_i(i$$=$$1,2,3)$ specified by Eq. (\ref{ext1}). $\rho_1,\rho_2,\rho_3$ corresponding to any point in the shaded region can be used for running $\mathcal{N}_4.$}}
			\label{fig3}
		\end{figure}
	\end{center}
\subsection{$\mathbf{Q}$ in terms of fixed amount of violation of Eq.(\ref{ineqs})}
Let $\rho_i$$=$$\varrho$ be such that the correlation tensor $T$ of $\varrho$ is $\textmd{diag}(t,t,t).$ On using such $\varrho,$ let $\textbf{V}$ amount of violation of the trilocal inequality(Eq.(\ref{ineqs})) be observed in $\mathcal{N}_4.$ Hence:
\begin{eqnarray}\label{extr1}
	\sqrt{2t^{\frac{2}{3}}}&=&2^{\frac{1}{3}}+\mathbf{V}\nonumber\\
	\Rightarrow t&=&\frac{(\mathbf{V}+2^{\frac{1}{3}})^3}{2\sqrt{2}}
\end{eqnarray}
Using Eq.(\ref{extr1}), QBER $\mathbf{Q}^{'}$(Eq.(\ref{cr8i})) takes the form:
\begin{eqnarray}\label{extr2}
	\mathbf{Q}^{'}&=&1-\frac{(1+\frac{(\mathbf{V}+2^{\frac{1}{3}})^3}{2\sqrt{2}})^3}{2^3}
\end{eqnarray}
Eq.(\ref{extr2}) points out that $\mathbf{Q}^{'}$ decreases with increase in $\mathbf{V}$(see fig.\ref{vq}).
	\begin{center}
	\begin{figure}
		\includegraphics[width=3in]{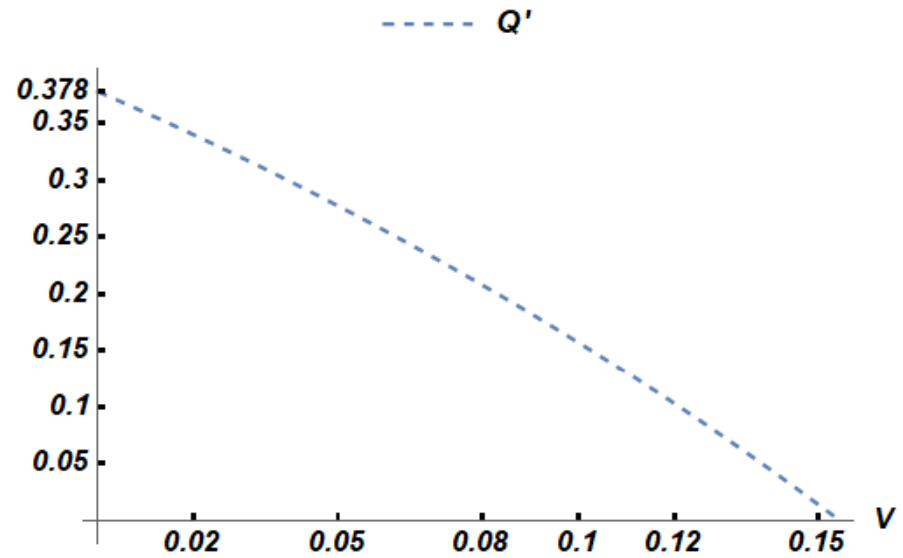} 
		\caption{\emph{$\mathbf{V}$ vs $\mathbf{Q}^{'}$ plot when three identical two qubit states with correlation tensor $\textmd{diag}(t,t,t)$ are used in $\mathcal{N}_4.$ Horizontal axis represents any fixed amount of violation($\mathbf{V}$) observed in $4^{th}$ step of the protocol. In case of no violation, i.e., $\mathbf{V}$$=$$0,$ QBER(minimized over all measurement settings) turns out to be $0.3781.$ $\mathbf{Q}^{'}$ gradually decreases and ultimately becomes $0$ in case Eq.(\ref{ineqs}) is violated maximally.}}
		\label{vq}
	\end{figure}
\end{center}
\subsection{Not All FNN States Are Useful}
For $\rho_1,\rho_2,\rho_3$ to be useful in the protocol both the validation criteria $\mathcal{C}_{\mathcal{N}_1}$ and $\mathcal{C}_{\mathcal{N}_2}$(or $\mathcal{C}_{\mathcal{N}_2^{'}}$) need to be satisfied. When $\rho_1$$=$$\rho_2$$=$$\rho_3$$=$$\varrho$(identical), the state must satisfy Eq.(\ref{cr10}) and Eq.(\ref{cr9i}) simultaneously. State corresponding to any point $\mathbf{P}$ lying below the tangent line($\mathbf{T}$) but outside the circle's positive quadrant($\mathbf{C}_+$) is $FNN$ but does not satisfy Eq.(\ref{cr9i}). Consequently, when any such state is used, $\mathcal{N}_4$ will be aborted in the sifting stage(see Fig.\ref{fig2}). \\
	In case $\rho_1,\rho_2,\rho_3$ are non identical, the state parameters must satisfy Eq.(\ref{cr7}) and Eq.(\ref{cr9}). However, there exist states abiding by Eq.(\ref{cr7}) but violate Eq.(\ref{cr9}). Such states generate fully network nonlocal correlations but still cannot be used to execute $\mathcal{N}_4$(see sub-fig.(ii) of Fig.\ref{fig2} and also Fig. \ref{exfig}). Violation of trilocal inequality thus acts as a necessary condition but does not suffice to guarantee successful run of $\mathcal{N}_4.$ 
	\begin{center}
		\begin{figure}
			\includegraphics[width=3.1in]{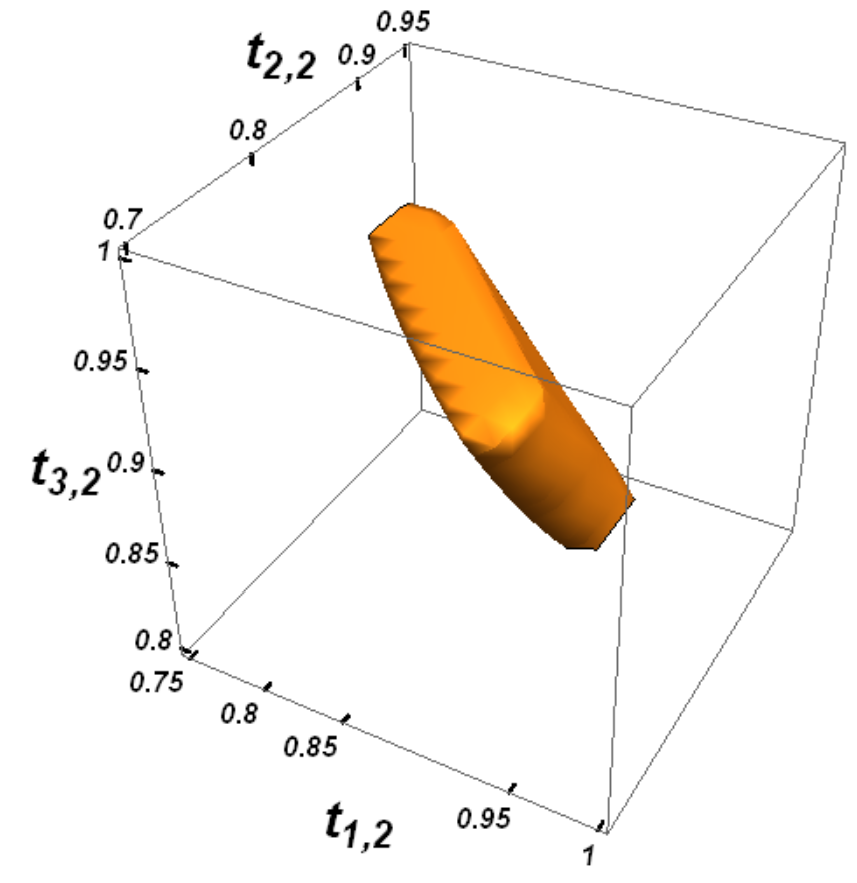} 
			\caption{\emph{Shaded region is a subspace formed by $2^{nd}$ largest singular values of $\rho_1,\rho_2,\rho_3$ specified by $(t_{1,1},t_{2,1},t_{3,1})$$=$$(0.91,0.94,0.93).$ Any such states generate detectable full network nonlocality but cannot be used for running $\mathcal{N}_4.$}}
			\label{exfig}
		\end{figure}
	\end{center}
	\subsection{Mis-classification of Useful States}
	From previous discussions, it is clear that the trusted parties can use either $\mathcal{C}_{\mathcal{N},2}$ or $\mathcal{C}_{\mathcal{N},2}^{'}$ to check validity of the protocol in the sifting stage. Choice in between these two alternatives is crucial in characterizing two-qubit states in context of their use in designing QKD protocol. \\
	Consider an ideal QKD scenario where only trusted parties are present and $\rho_i$$=$$\varrho,\forall i$. So, after maximization over all measurement settings, $\mathbf{H}$(Eq.(\ref{cr6})) is given by:
	\begin{eqnarray}
		\mathbf{H}_{\textmd{\tiny{Max over all MUBs}}}&=&(t_1+t_2+2)^3.
	\end{eqnarray}
	However, let the parties do not know that all the states involved in the protocol are identical. Consequently, they use $\mathcal{C}_{\mathcal{N},2}$ as second validation criterion. Let $\varrho$ be such that it satisfies: 
	\begin{eqnarray}\label{pecu1}
		t_1^2+t_2^2&>& 2^{\frac{2}{3}}
	\end{eqnarray}
	\begin{eqnarray}\label{pecu2}
		t_1+t_2&\leq&2^{\frac{4}{3}}(3+(2^{\frac{2}{3}}-1)^{\frac{3}{2}})^
		{\frac{1}{3}}-2
	\end{eqnarray}
	and
	\begin{eqnarray}\label{pecu3}
		t_1+t_2&>&2^{\frac{5}{6}} 
	\end{eqnarray}
	As $\varrho$ satisfies Eq.(\ref{pecu1}), $\mathcal{N}_4$ passes the first validation check. As Eq.(\ref{pecu2}) holds and the parties are using $\mathcal{C}_{\mathcal{N},2}$ as second checking criterion, so $\mathcal{N}_4$ cannot pass the second  validation check. The protocol is thus ultimately aborted indicating $\varrho$ to be useless for designing QKD protocol. However $\varrho$ satisfies Eq.(\ref{pecu3}). So, if $\mathcal{C}_{\mathcal{N},2}^{'}$ was used then the protocol would pass the second validation check also and thus run successfully. So here $\varrho$ gets mis-classified as not fit for designing $\mathcal{N}_4$(see Fig.\ref{fig4}). To this end it may be noted that for designing any QKD protocol ensuring security of the protocol must be given the top priority. From that perspective using more stringent checking criterion is acceptable even if some noisy entangled states get discarded unnecessarily. 
	\begin{center}
		\begin{figure}
			\includegraphics[width=3.1in]{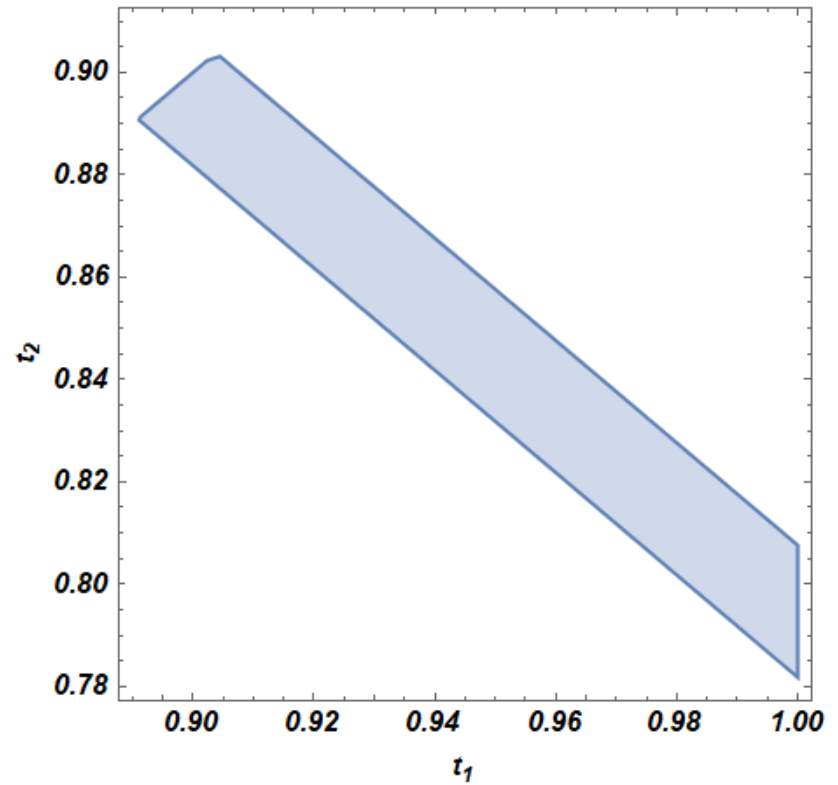} 
			\caption{\emph{The figure provides a subspace in two-dimensional space formed by largest two singular values of correlation tensor of arbitrary two-qubit state $\varrho.$ For any point $(t_1,t_2)$ lying in the shaded region, corresponding noisy entangled state satisfies both $\mathcal{C}_{\mathcal{N},1}$ and $\mathcal{C}_{\mathcal{N},2^{'}}$ but violates $\mathcal{C}_{\mathcal{N},2}.$ Consequently, any such state will be rejected as useless state if one uses $\mathcal{C}_{\mathcal{N},2}$ as a second checking criterion in $\mathcal{N}_4.$}}
			\label{fig4}
		\end{figure}
	\end{center}
\section{Eve-Aware Analysis}\label{new1}
Now any untrusted party can have control over each of the three independent sources $\mathbf{S}_1,\mathbf{S}_2,\mathbf{S}_3.$ For further analysis let it be assumed that Eve applies the purification strategy(see subsec.\ref{pre}) while still respecting the trilocality(source independence) assumption in the network. Precisely, Eve prepares a purification $|\psi_P^{(i)}\rangle\langle \psi_P^{(i)}|$ for $\varrho_i,\forall i$$=$$1,2,3.$ Such a strategy allows Eve to extract the maximum possible side information about each source\cite{dw1} separately without breaking the trilocal network structure maintained between the trusted parties. The overall quantum state in the network is then given by:
\begin{eqnarray}
	\varrho^{'}&=&\otimes_{i=1}^3 (|\psi_P^{(i)}\rangle\langle \psi_P^{(i)}|)
\end{eqnarray}
Clearly, $\forall i,$ $\varrho_i$$=$$\textmd{Tr}_E|\psi_P^{(i)}\rangle\langle \psi_P^{(i)}|.$ \\Now $\forall i,$ Eve can measure its subsystem of $|\psi_P^{(i)}\rangle\langle \psi_P^{(i)}|$ in any of the three possible single-qubit MUBs. Let $I_{i,j}$ denote its mutual information obtained from $|\psi_P^{(i)}\rangle\langle \psi_P^{(i)}|$ while measuring in $j^{th}$ basis($j$$=$$1,2,3$). Now $\forall j$$\neq$$k,$ measurement in $j^{th}$ and $k^{th}$ bases being independent and also due to the independence of the $|\psi_P^{(1)}\rangle\langle \psi_P^{(1)}|,|\psi_P^{(2)}\rangle\langle \psi_P^{(2)}|,|\psi_P^{(3)}\rangle\langle \psi_P^{(3)}|,$ the overall mutual information obtained by Eve from the network is provided by the additivity of $I_{i,j}$\cite{nie}:
\begin{eqnarray}
	I(A_1:E)&=&\frac{1}{9}\sum_{i,j=1}^{3}I_{i,j}\nonumber\\
	&\leq&\frac{1}{9}\sum_{i,j=1}^{3}\chi_{i,j}
\end{eqnarray}
where $\chi_{i,j}$ is the Holevo quantity bounding $I_{i,j}.$ The factor $\frac{1}{9}$ is considered for averaging over $I_{i,j},\forall i,j.$ Due to independence over both $i,j,$ the asymptotic $1$-way key rate for $\mathcal{N}_4$ is given by:
\begin{eqnarray}
	R_{\mathcal{N}_4}&=&\sum_{i,j=1}^3R^{(i,j)}\nonumber\\
	&=&\frac{1}{9}\sum_{i=1}^3\sum_{j=1}^{3}(I(A_1^{(i)}:A_{i+1})-\chi_{i,j})_{j^{th}\small{MUB}}
\end{eqnarray}
$\forall i$$=$$1,2,3,$ $A_1^{(i)}$ denotes the classical variable corresponding to projective measurement of $A_1$ on its subsystem of $\rho_i.$ \\
A $1$-way secret key can be distilled from $\mathcal{N}_4$ if:
\begin{equation}\label{kd1}
	R_{\mathcal{N}_4}>0.
\end{equation}
Eve can access at most $\frac{1}{9}\sum_{i,j=1}^{3}\chi_{i,j}$ bits of side information from $\mathcal{N}_4.$ 
\subsection{Illustration}
Let $\forall i,\mathbf{S}_i$ generate a pure state:
\begin{eqnarray}\label{pures}
	\rho_i&=&(\cos(x_i)|00\rangle+\sin(x_i)|11\rangle)\cdot(\cos(x_i)\langle 00|\nonumber\\
	&+&\sin(x_i)\langle 11|).
\end{eqnarray}
$\mathcal{N}_4$ will be considered for secret key generation if it satisfies both $\mathcal{C}_{\mathcal{N},1}$ and $\mathcal{C}_{\mathcal{N},2}.$
Using key distillation criterion(Eq.(\ref{kd1})), it is observed that there exist state parameters for which secret key can be distilled from $\mathcal{N}_4$(see Fig.\ref{dist1}) even when Eve gets access to side information. 
\begin{center}
	\begin{figure}
		\includegraphics[width=3.1in]{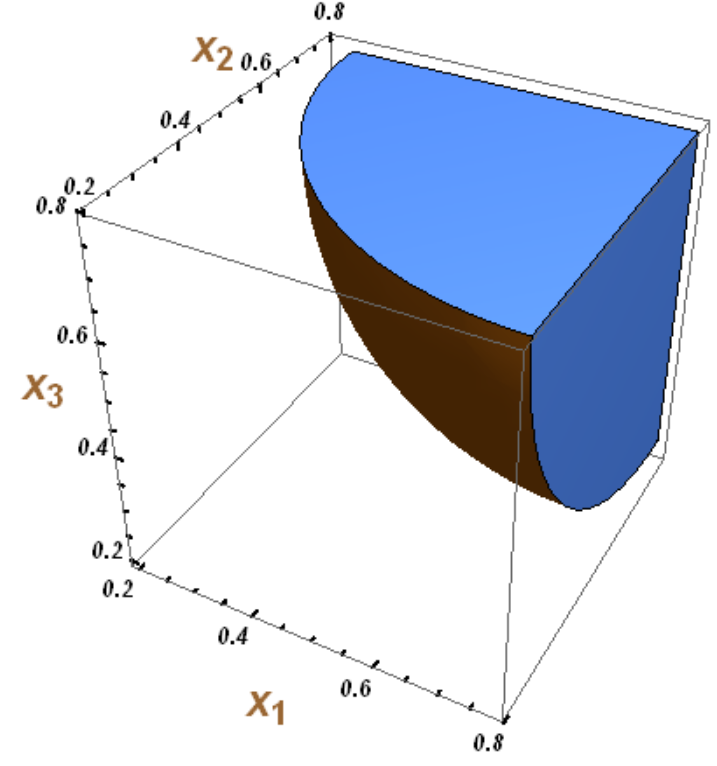} 
		\caption{\emph{Shaded region gives the subspace in parameter space $(x_1,x_2,x_3)$ for which $1$-way secret key can be distilled from $\mathcal{N}_4.$}}
		\label{dist1}
	\end{figure}
\end{center}
For a numerical instance, let $x_1$$=$$0.65,x_2$$=$$x_3$$=$$0.7$ Both  $\mathcal{C}_{\mathcal{N},1}$ and $\mathcal{C}_{\mathcal{N},2}$ are satisfied and hence $\mathcal{N}_4$ is not aborted.\\
Here $\frac{1}{9}\sum_{i,j=1}^{3}\chi_{i,j}$$=$$0.134.$ Eve can thus gain at most $0.134$ bits of side information and the secret key distillation rate turns out to be $0.605$ bits. 
\section{Detection Efficiency of $\mathcal{N}_4$}\label{new2}
Let Eve has control over the measurement devices of the trusted parties in a way that sometimes the devices fail to register any incoming qubit. As discussed in sec.\ref{pre}, such event of no detection($\phi,$say) may be referred to as \textit{'no-click event}'. Let $\eta$ be the detection efficiency of each party's detector. $\forall i$, let $A_i$ retain $\phi$ such that:
\begin{eqnarray}\label{detef5}
	A_i \textmd{ assign } \phi&\rightarrow&1,\,\forall i=2,3,4\nonumber\\
	A_1 \textmd{ assign } \phi&\rightarrow&(1,1,1).
\end{eqnarray}
Now, by Eq.(\ref{detef5}), re-labeled output set of the outputs remain the same as in ideal case with outputs of extreme parties lying in $\{0,1\}$ and the set of raw output strings of $A_1$ given by $\{(i,j,k)|i,j,k=0,1\}.$ \\
For above retention strategy(Eq.(\ref{detef5})), the correlator terms in trilocal inequality(Eq.(\ref{ineqs})) gets modified as follows:
\begin{eqnarray}\label{detef6}
	\langle A_{(1)}^{(i)}A_{x_2}A_{x_{3}}A_{x_{4}}\rangle_{\eta}=\eta F_1+(1-\eta)F_2,&&\\
	F_1=(\eta^3 \langle A_{(1)}^{(i)}A_{x_2}A_{x_{3}}A_{x_{4}}\rangle_{id}+(1-\eta)^3\langle A_{(1)}^{(i)}\rangle&&\nonumber\\
	+\eta^2(1-\eta)(\langle A_{(1)}^{(i)}A_{x_{3}}A_{x_{4}}\rangle+\langle A_{(1)}^{(i)}A_{x_{2}}A_{x_{4}}\rangle+&&\nonumber\\
	\langle A_{(1)}^{(i)}A_{x_{2}}A_{x_{3}}\rangle)+\eta(1-\eta)^2(\sum_{j=2}^4\langle A_{(1)}^{(i)}A_{x_{j}}\rangle))&&\nonumber\\
	\textmd{\small{and }} F_2= (\eta^3 \langle A_{x_2}A_{x_{3}}A_{x_{4}}\rangle_{id}+(1-\eta)^3+&&\nonumber\\
	\eta^2(1-\eta)(\langle A_{x_{3}}A_{x_{4}}\rangle+\langle A_{x_{2}}A_{x_{4}}\rangle+\langle A_{x_{2}}A_{x_{3}}\rangle)&&\nonumber\\
	+\eta(1-\eta)^2(\sum_{j=2}^4\langle A_{x_{j}}\rangle))&&\nonumber
\end{eqnarray} 
As the binary valued relabeled raw and hence classically post-processed(in case of $A_1$) outputs remain invariant, the bound of the trilocal inequality remains unaltered:
\begin{eqnarray}\label{detef7}
	&&\frac{1}{2}\sum_{i=1}^{4}|J_{i}(\eta)|^{\frac{1}{3}}\leq  2^{\frac{1}{3}},\,  \textmd{\small{where}}\\
	&& J_i(\eta)=\frac{1}{2^3}\sum_{x_2,x_3,x_{4}}(-1)^{g_i(x_2,x_3,x_{4})}\langle A_{(1)}^{(i)}A_{x_2}A_{x_{3}}A_{x_{4}}\rangle_{\eta}\\
	&&\textmd{\small{with }} g_i \textmd{\small{ given by Table.\ref{table:ta6} and }}	\langle A_{(1)}^{(i)}A_{x_2}A_{x_{3}}A_{x_{4}}\rangle_{\eta} \textmd{ \small{by Eq.(\ref{detef6}).}} \nonumber
\end{eqnarray}
\textit{Influence of $\eta$ over $\mathcal{N}_4$:}
The first validation check in $\mathcal{N}_4$ now relies upon violation of Eq.(\ref{detef7}):\textit{If Eq.(\ref{detef7}) is violated then next step of $\mathcal{N}_4$ is executed. Otherwise it is aborted.}\\ Clearly, such a check now depends on the detection efficiency parameter $\eta.$ Smaller value of $\eta$  indicates that the detectors fail to detect most of the systems sent to the trusted parties. Now, as Eve has control over these devices, low $\eta$ thus implies greater chance for Eve to fake or enhance an apparent violation of Eq.(\ref{detef7}). So higher the value of $\eta$ greater is the resistance offered by $\mathcal{N}_4$ towards detection loophole. Consequently in $\mathcal{N}_4$ with large $\eta,$ the untrusted party gets lesser chance to use detector losses in the first validation check of the protocol.
\subsection{Illustration}
Let each of the source distribute a pure state(Eq.(\ref{pures})) in $\mathcal{N}_4.$ The trusted parties, using the detectors with efficiency parameter $\eta$ execute step 2(ii)(for $A_1$) and 2(b)(for $A_2$). For suitable choices of measurement settings, corresponding measurement correlations violate Eq.(\ref{detef7}) if the following condition holds:
\begin{eqnarray}
	\sqrt{2}\eta^{\frac{4}{3}}&>&2^{\frac{1}{3}}
\end{eqnarray}
So, Eq.(\ref{detef7}) is violated if $\eta$$>$$0.917.$ This in turn indicates that the correlations,shared among $A_1,A_2,A_3,A_4$ exhibit full network nonlocality with relatively few `no-click' events, diminishing $\mathcal{N}_4^{'}$s vulnerability to the detection loophole.
\section{Comparing $\mathcal{N}_4$ With Protocol Relying on Bell-CHSH Violation}\label{cpmp}
From discussions above it is clear that $\mathcal{N}_4$ is an entanglement assisted QKD protocol that relies on detection of genuine form of four-partite network correlations. In this context it becomes imperative to explore the network based protocol when its security relies upon detection of bipartite nonlocal correlations in each of the pairs $(A_1,A_2),$ $(A_1,A_3)$ and $(A_1,A_4)$ sharing a two-qubit state. For that consider a protocol $\mathbf{N}_4$(say) that relies upon violation of Bell-CHSH inequality instead of Eq.(\ref{ineqs}). In $\mathbf{N}_4$ four legitimate parties $A_1,A_2,A_3,A_4$ are involved in same pattern as that in $\mathcal{N}_4.$ Analogous to $\mathcal{N}_4$ this protocol will finally generate a block-structured secure key with block length $3$ that will be shared between $A_1,A_2,A_3,A_4.$ Almost all the steps of $\mathbf{N}_4$ are same as that of $\mathcal{N}_4.$ Modifications are as follows:
\begin{itemize}
	\item In step (2ii) $\forall i$$=$$1,2,3,$ for each of remaining $n-h_1$ copies of $\rho_i$, $A_1$ performs single-qubit projective measurements randomly in any one of two arbitrary directions $\vec{b}_{i,1}.\vec{\sigma},\vec{b}_{i,2}.\vec{\sigma}$ on the single qubit of $\rho_i.$
	\item In step (4) the parties communicate among themselves their outputs resulting in second step of measurement stage(steps 2(ii) and 2(b)). Bipartite correlations $P(a_1,a_{i+1}|x_1,x_{i+1})$ from each of the pair  $(A_1,A_{i+1})(i$$=$$2,3,4)$ are collected to test the Bell-CHSH inequality. Here, $\forall i=1,2,3,4$ $(x_i,a_i)$ denotes the input-output pair of $i^{th}$ trusted party $A_i.$\\
	If violation of Bell-CHSH inequality is observed for each of the three pairs $(A_1,A_2),$ $(A_1,A_3)$ and $(A_1,A_4)$ then the next step of the protocol is executed. Otherwise, the protocol is aborted.
\end{itemize}
\subsection{Minimizing QBER In $\mathbf{N}_4$}
	Sifting stage of $\mathbf{N}_4$ remaining same as in $\mathcal{N}_4,$ QBER generated in the protocol is given by Eq.(\ref{cr4}). Minimizing QBER is equivalent to maximizing $\mathbf{H}$ given by Eq.(\ref{cr5}). Now, as argued before, maximizing $\mathbf{H}$ with respect to MUBs is unconstrained. On being maximized(with respect to MUBs) it is given by Eq.(\ref{cr6}). Next level of maximization is constrained. Unlike that in $\mathcal{N}_4,$ here maximization of $\mathbf{H}$ with respect to state parameters is performed under the constraint that there is no Bell-CHSH violation in at least one of the three pairs($(A_1,A_{i+1})$):
	\begin{eqnarray}\label{cr14}
		t_{i,1}^2+t_{i,2}^2&\leq& 1,\,\,\textmd{\small{for at least one }}i\in\{1,2,3\}
	\end{eqnarray}
	Let $\mathbf{Q}_0^{'}$ be the critical value of QBER obtained in $\mathbf{N}_4.$ In absence of detectable nonlocality among at least one of the pairs of parties, $\mathbf{Q}$ can never be made less than $\mathbf{Q}_0^{'}.$ However, it can be reduced further once each of the three pairs of trusted parties in $\mathbf{N}_4$ detect nonlocal correlations.\\
	Above criterion(Eq.\ref{cr14}) puts restriction over $t_{i,k}.$
	Theorem below provides $\mathbf{Q}_0^{'}$ for $\mathbf{N}_4.$
	\begin{theorem}\label{theo2}
		In the network based $4$-party QKD protocol $\mathbf{N}_4,$ involving three two-qubit states, QBER generated cannot be less than $\mathbf{Q}_0^{'}$$=$$1-(\frac{1+\sqrt{2}}{2\sqrt{2}})^c$ where $c$($1$$\leq$$c$$\leq$$3$) denote number of pairs of central and extreme parties that do not observe Bell-CHSH violation. 
	\end{theorem}
	\textbf{Proof:}See Appendix.C.\\
	For $\mathbf{N}_4,$ theorem.\ref{theo2} thus provides a threshold value of QBER which is obtained under constraint that Bell-CHSH violation is not observed in $c$ number of pairs of parties in Step.$4$. \\
$\mathbf{Q}_0^{'}$ depends on how many of the three pairs of one extreme and one central party does not show Bell-CHSH violation(see Table.\ref{table:cases}). It turns out that $\mathbf{Q}_0^{'}$ is monotonic increasing with $c.$ Such a dependency of $\mathbf{Q}_0^{'}$ on number of pairs not showing Bell-CHSH violation will next be used to frame validation criterion.
\begin{table}
	\caption{Different possible cases of observing Bell-CHSH violation in $\mathbf{N}_4$ and corresponding value of $\mathbf{Q}_0^{'}.$}
	\begin{center}
		\begin{tabular}{|c|c|}
			\hline
			\small{Number of cuts showing}&\small{Approx value of} \\
			\small{CHSH violation}&  $\mathbf{Q}_0^{'}$\\
			\hline
		 nil& $0.37814$\\
			\hline
		 one& $0.27145$\\
			\hline
		two&$0.14645$\\
			\hline
		\end{tabular}
	\end{center}
	\label{table:cases}
\end{table}
	\subsection{Criteria Necessary For Validating $\mathbf{N}_4$}\label{special}
In $\mathbf{N}_4,$ validation check is done in the following two steps:
	\begin{itemize}
		\item \textit{First Check:} In $4^{th}$ step violation of Bell-CHSH inequality tested for each pair $(A_1,A_2),$ $(A_1,A_3),$ $(A_1,A_4).$ Here the check relies upon the fact that each of the three pairs of trusted parties $(A_1,A_{i+1})$ share detectable nonlocal correlations. \\
		State $\rho_i$ is shared in between $(A_1,A_{i+1}).$ It may happen that even in presence of Eve, Bell-CHSH violation is obtained from some of $\rho_1,\rho_2,\rho_3$. So to make the validation criterion more stringent violation for each of $\rho_1,\rho_2,\rho_3$ is set as a mandate. Now, any two-qubit state $\rho_i$ is Bell-CHSH nonlocal if it satisfies\cite{horo}: 
		%
		\begin{eqnarray}\label{cr15}
			t_{i,1}^2+t_{i,2}^2&>&1.
		\end{eqnarray}
		Eq.(\ref{cr15}) acts as a checking criterion($\mathcal{C}_{\mathbf{N},1},$say) for $\mathbf{N}_4.$\\
		$\mathcal{C}_{\mathbf{N},1}:$ \textit{in $4^{th}$ step if Eq.(\ref{cr15}) is satisfied $\forall i$$=$$1,2,3$ then next step of the protocol is executed. Otherwise it is aborted.} 
		\item \textit{Second Check:} In $5^{th}$ step of the protocol second validation check is provided using $\mathbf{Q}_0^{'}$ from Theorem.\ref{theo2}. QBER can be reduced below $\mathbf{Q}_0^{'}$ when Bell-CHSH violation is observed from each of $\rho_1,\rho_2,\rho_3:$
		\begin{eqnarray}
			\mathbf{Q}&<&\mathbf{Q}_0^{'}.
		\end{eqnarray}
		Let not all three states used in $\mathbf{N}_4$ be identical. As violation needs to be observed in each of the three pairs, so $\mathbf{Q}_0^{'}$ provided in $3^{rd}$ of Table.\ref{table:cases} is to be considered here:
		\begin{eqnarray}\label{cr16}
			\mathbf{Q}&<&\mathbf{Q}_0^{'}=0.14645\nonumber\\
			\Rightarrow\sum_{\substack{j_i=1,2\\\forall i=1,2,3}}\Pi_{i=1}^3(1+t_{i,j_i})&>&16\sqrt{2}(1+\sqrt{2}).
		\end{eqnarray}
		Let $\mathcal{C}_{\mathbf{N},2}$ denote the second validation criterion based on Eq.(\ref{cr16}).\\
		$\mathcal{C}_{\mathbf{N},2}:$\textit{if Eq.(\ref{cr16}) is satisfied in $5^{th}$ step then $\mathbf{N}_4$ is used to generate secure key. Otherwise it is aborted.}\\
		However, if all the three states are identical, $\rho_i$$=$$\varrho,\forall i$ then Eq.(\ref{cr16}) gets modified:
		\begin{eqnarray}\label{cr16i}
			t_1+t_2&>&\sqrt{2},\,\,\textmd{\small{Using Eq.(\ref{cr6i})}}.
		\end{eqnarray}
		In that case, second validation criterion($\mathcal{C}_{\mathbf{N},2}^{'},$say) is similar to $\mathcal{C}_{\mathbf{N},2}$ with only Eq.(\ref{cr16}) replaced by Eq.(\ref{cr16i}).
	\end{itemize}
	\par Now that two different network based QKD protocols have been designed, it becomes pertinent to compare their efficiency. In following subsection the two protocols are compared in terms of the security provided by them. 
	\subsection{$\mathcal{N}_4$ Offers More Stricter Validity Checks Than $\mathbf{N}_4$}
	Both the protocols designed here rely upon violation of some correlator based inequalities for framing validation criteria to detect presence of malicious party. Also the extent up to which the QBER generated in the protocols can be reduced depend on these inequalities. A comparison of both first and second validation criteria will aid in comparing security of these protocols to detect presence of eavesdropper.
	\subsubsection{$\mathcal{C}_{\mathcal{N},1}$ Versus $\mathcal{C}_{\mathbf{N},1}$} $\mathcal{C}_{\mathcal{N},1}$ involves violation of trilocal inequality(Eq.(\ref{ineqs})) by four-partite correlations whereas $\mathcal{C}_{\mathbf{N},1}$ depends upon violation of Bell-CHSH by bipartite correlations in all three possible pairs of central and an extreme party. Now, as already pointed out before, only genuine form of four partite network nonlocal correlations can violate Eq.(\ref{ineqs}). Again such form of network nonlocality cannot exist if at least one pair $(A_1,A_{i+1})$ share local correlations\cite{bilo6,birev}. Even if only one of  $\rho_1,\rho_2,\rho_3$ does not violate Bell-CHSH whereas both the others show maximal violation then also such states cannot violate Eq.(\ref{ineqs}). For instance, let $\rho_1,\rho_2,\rho_3$ be such that $\rho_1,\rho_2$(say) show maximum quantum violation whereas the remaining state($\rho_3$) does not violate Bell-CHSH inequality:
	\begin{eqnarray}\label{comp1}
		t_{i,1}^2+t_{i,2}^2&=&1,\,\,i=1,2\nonumber\\
		t_{3,1}^2+t_{3,2}^2&\leq&1.
	\end{eqnarray}
	Thus above states violate $\mathcal{C}_{\mathbf{N},1}$ and cannot be used to design $\mathbf{N}_4.$\\
	Again $t_{i,j}$ satisfying Eq.(\ref{comp1}) will always satisfy the following relation:
	\begin{eqnarray}
		t_{3,1}^{\frac{2}{3}}+t_{3,2}^{\frac{2}{3}}&\leq& 2^{\frac{2}{3}}
	\end{eqnarray}
	\begin{center}
		\begin{figure}
			\includegraphics[width=3.1in]{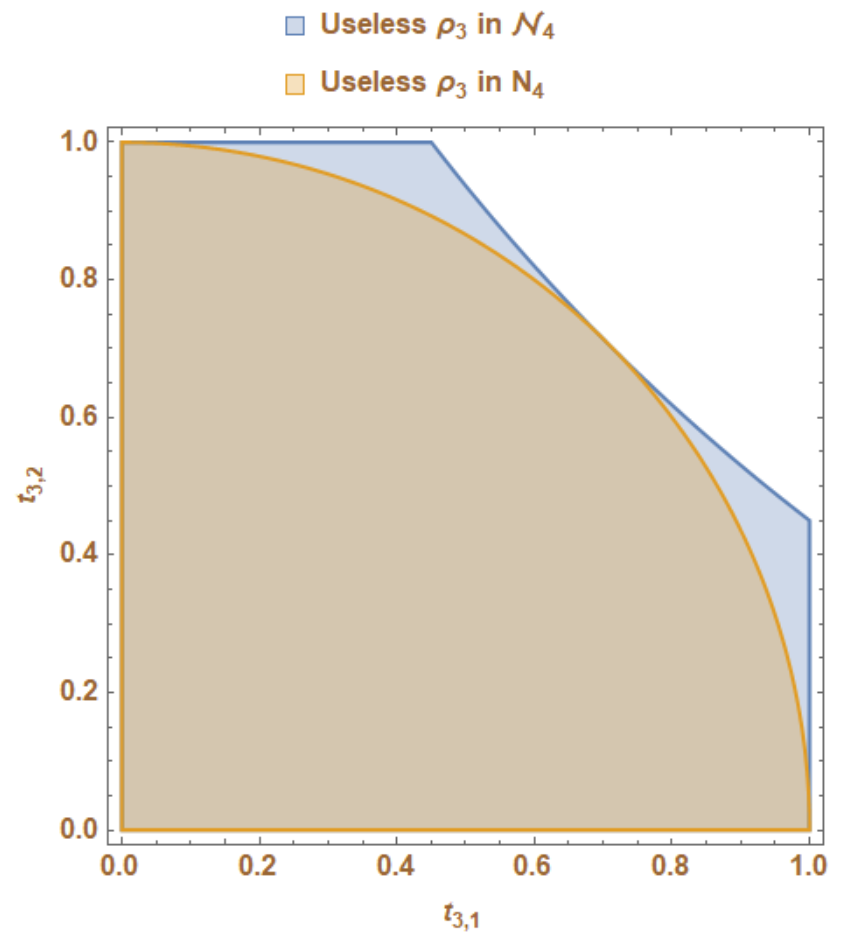} 
			\caption{\emph{Figure justifies $\mathcal{R}_1$ considering an arbitrary two-qubit state $\rho_3$ used with two maximally entangled two-qubit states.}}
			\label{fig7}
		\end{figure}
	\end{center}
	Thus $\mathcal{N}_4$ also cannot be designed using these states. It is thus clear that states violating $\mathcal{C}_{\mathbf{N},1}$ will also violate $\mathcal{C}_{\mathcal{N},1}$(\ref{fig7}). Consequently, the following result($\mathcal{R}_1$,say) holds.\\
	$\mathcal{R}_1:$\textit{using any set of $\rho_1,\rho_2,\rho_3,$ if $\mathbf{N}_4$ is aborted in $4^{th}$ step then that set of two-qubit states will also result in abortion of $\mathcal{N}_4$ in $4^{th}$ step.}\\
	However, the reverse of $\mathcal{R}_1$ is not always true. Precisely, if some $\rho_1,\rho_2,\rho_3$ satisfy $\mathcal{C}_{\mathbf{N},1}$ then they do not necessarily satisfy $\mathcal{C}_{\mathcal{N},1}.$ This is because for any set of three two-qubit states  existence of bipartite nonlocal correlations in each of three pairs of central and extreme party($(A_1,A_{i+1})\forall i$) does not ensure generation of full network nonlocality among $A_1,A_2,A_3,A_4.$ \\
	For instance, consider the following two-qubit states:
	\begin{eqnarray}\label{ext3}
		\rho_1\textmd{ \small{with} }T_1&=&\textmd{diag}(0.92,t_{1,2},t_{1,3})\nonumber\\
		\rho_2\textmd{ \small{with} }	T_2&=&\textmd{diag}(0.94,t_{2,2},t_{2,3})\nonumber\\
		\rho_3\textmd{ \small{with} }	T_3&=&\textmd{diag}(0.95,t_{3,2},t_{3,3})
	\end{eqnarray}
	There exist parameters $t_{1,2},t_{2,2},t_{3,2}$ for which above states satisfy $\mathcal{C}_{\mathbf{N},1}$ but violate $\mathcal{C}_{\mathcal{N},1}$(see sub-fig.(i) in Fig.\ref{fig8}).\\
	\begin{center}
		\begin{figure}
			\begin{tabular}{c}
				\subfloat[ ]{\includegraphics[trim = 0mm 0mm 0mm 0mm,clip,scale=0.39]{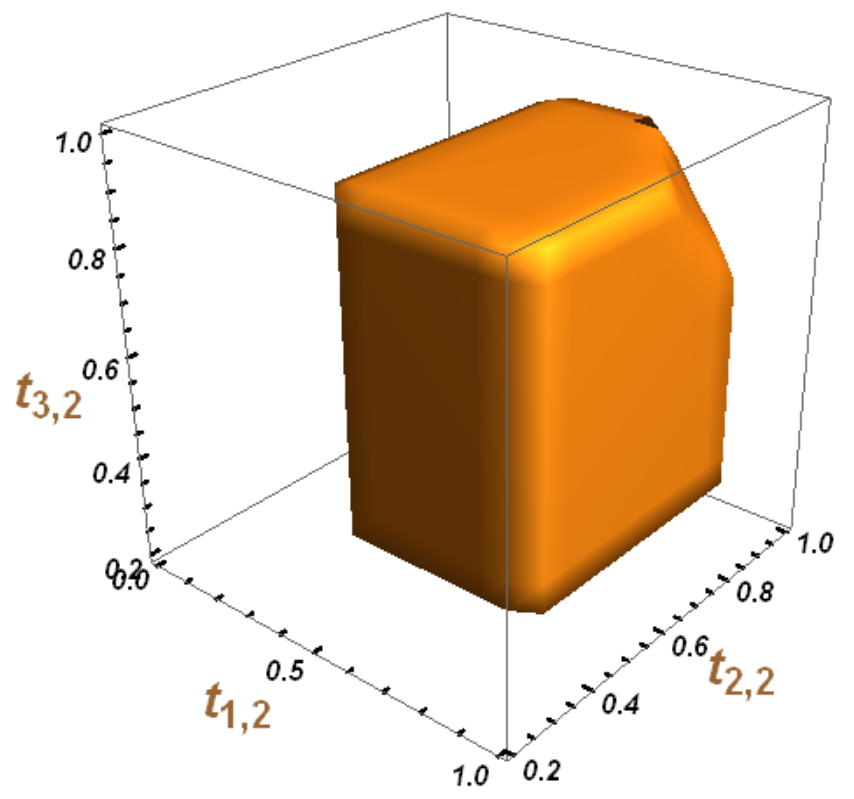}}\\
				\subfloat[]{\includegraphics[trim = 0mm 0mm 0mm 0mm,clip,scale=0.39]{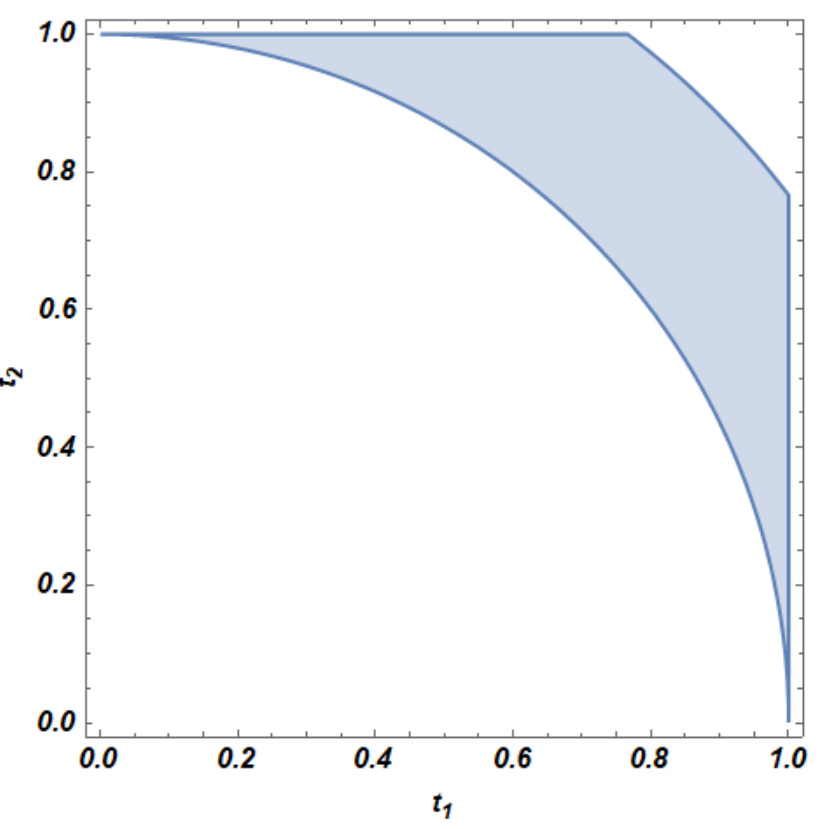}}\\
			\end{tabular}
			\caption{\emph{Shaded region provides instances of non-identical two-qubit states(Eq.(\ref{ext3})) in support of the result $\mathcal{R}_2.$ Similar implication is provided by sub-figure(ii) when $\rho_1,\rho_2,\rho_3$ are identical($\varrho$).}} 
			\label{fig8}
		\end{figure}
	\end{center}
	Thus for some given $\rho_1,\rho_2,\rho_3$ it may happen that $\mathbf{N}_4$ passes first validation check and hence eavesdropper does not get detected up to $4^{th}$ step of $\mathbf{N}_4.$ However for the same states, $\mathcal{N}_4$ gets aborted after first validation check only. Due to existence of such states and also based on $\mathcal{R}_1$, following result($\mathcal{R}_2,$say) thus becomes evident.\\
	$\mathcal{R}_2:$\textit{$\mathcal{C}_{\mathcal{N},1}$ is more stringent validation criterion compared to $\mathcal{C}_{\mathbf{N},1}.$}
	\subsubsection{Comparison Of Second Validation Check}
	Let the parties use the criterion $\mathcal{C}_{\mathcal{N}_2}$ in $\mathcal{N}_4$ and  $\mathcal{C}_{\mathbf{N}_2}$ in $\mathbf{N}_4.$ As discussed before, once the protocol passes the first validation check, QBER generated can be less than $13.7\%$(approx) in $\mathcal{N}_4.$ However, in $\mathbf{N}_4,$ QBER can be reduced below $14.6\%.$ This in turn points out the possibility that for some states QBER reduction will be more in the former protocol. \\
	Comparison of Eq.(\ref{cr9}) with that of Eq.(\ref{cr16}) clearly points out that there may exist states for which $\mathcal{C}_{\mathbf{N}_2}$ is satisfied whereas $\mathcal{C}_{\mathcal{N}_2}$ is violated. However the reverse is not possible. Precisely, for some states one may get:
	\begin{eqnarray}\label{pecu5}
		16\sqrt{2}(1+\sqrt{2})<\sum_{\substack{j_i=1,2\\\forall i=1,2,3}}\Pi_{i=1}^3(1+t_{i,j_i})&&\\
		\leq16(3+(2^{\frac{2}{3}}-1)^{\frac{3}{2}}).&&
	\end{eqnarray}
	Consider the following two-qubit states:
	\begin{eqnarray}\label{ext4}
		\rho_1\textmd{ \small{with} }T_1&=&\textmd{diag}(0.92,t_{1,2},t_{1,3})\nonumber\\
		\rho_2\textmd{ \small{with} }T_2&=&\textmd{diag}(0.91,t_{2,2},t_{2,3})\nonumber\\
		\rho_3\textmd{ \small{with} }T_3&=&\textmd{diag}(0.94,t_{3,2},t_{3,3})\nonumber\\
	\end{eqnarray}
	There exist parameters $t_{1,2},t_{2,2},t_{3,2}$(see Fig.\ref{fig9}) for which above states satisfy Eq.(\ref{pecu5}). 
	\begin{center}
		\begin{figure}
			\begin{tabular}{c}
				\subfloat[ ]{\includegraphics[trim = 0mm 0mm 0mm 0mm,clip,scale=0.39]{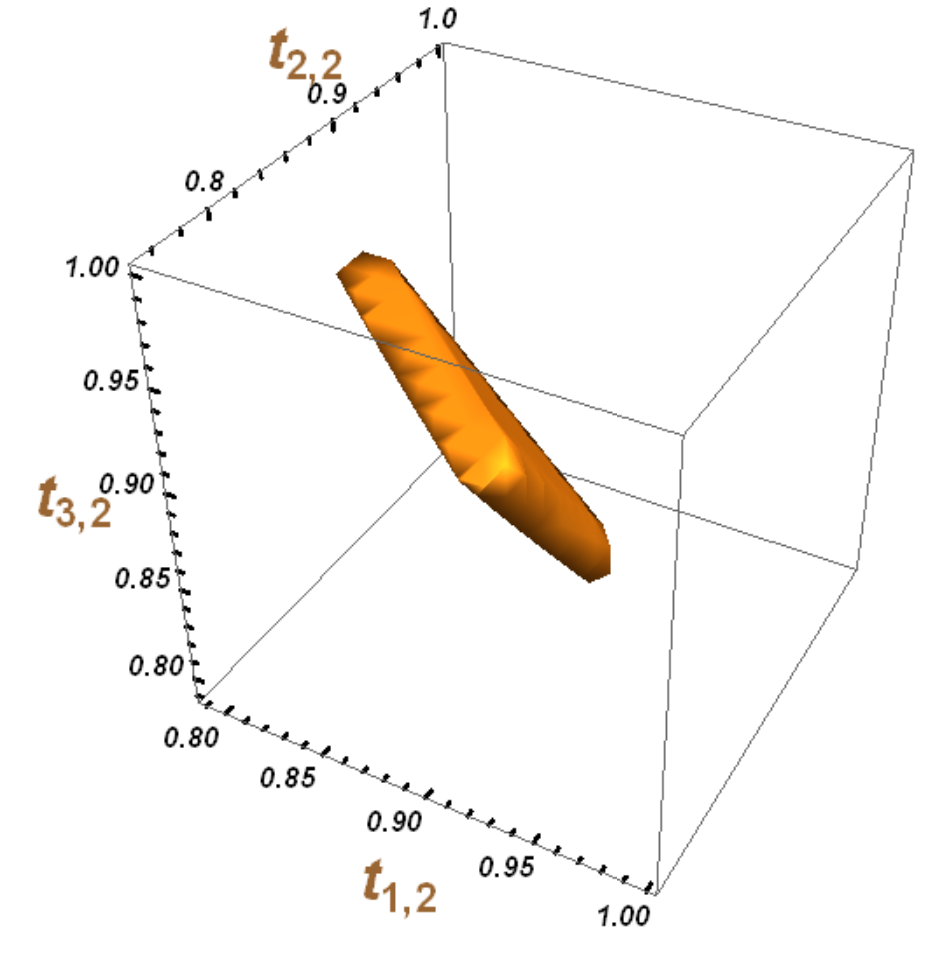}}\\
				\subfloat[]{\includegraphics[trim = 0mm 0mm 0mm 0mm,clip,scale=0.39]{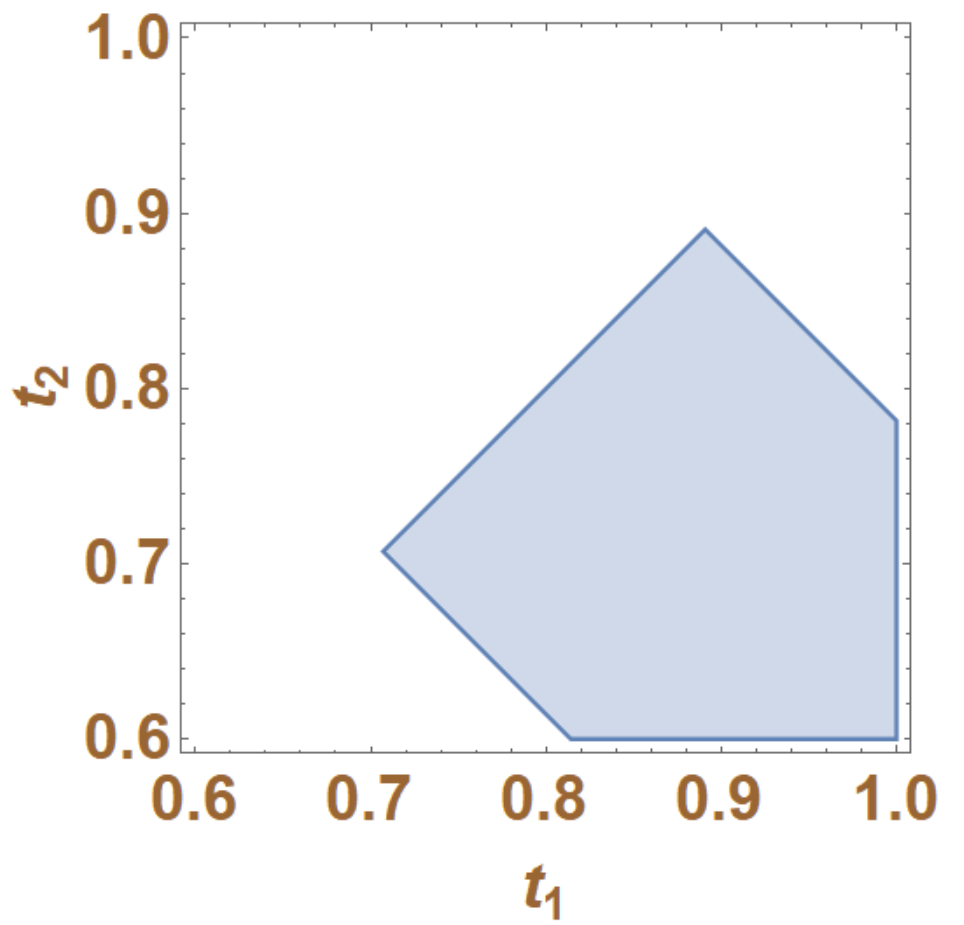}}\\
			\end{tabular}
			\caption{\emph{Sub-figure(i) provides instances of non-identical two-qubit states(specified by Eq.(\ref{ext4})) in support of the result $\mathcal{R}_3.$ Similarly, region in sub-figure(ii) also supports $\mathcal{R}_3$ in case $\rho_1,\rho_2,\rho_3$ are identical($\varrho$).}}
			\label{fig9}
		\end{figure}
	\end{center}
Now let both $\mathcal{N}_4$ and $\mathbf{N}_4$ involve identical states($\varrho$) and the parties use $\mathcal{C}_{\mathcal{N}_2}^{'}$ and $\mathcal{C}_{\mathbf{N}_2}^{'}$ in the respective protocols. Let $\varrho$ be such that the largest two singular values $t_1,t_2$ of $\varrho^{'}$s correlation tensor($T$) satisfy the following relation:
	\begin{eqnarray}\label{pecu8}
		\sqrt{2}<t_1+t_2&\leq&2^{\frac{5}{6}}\,\,\textmd{\small{Using Eqs.(\ref{pecu3},\ref{cr16i})}}
	\end{eqnarray}
	Clearly, $\varrho$ satisfying Eq.(\ref{pecu8}) satisfies Eq.(\ref{cr16i}). $\mathbf{N}_4$ using $\varrho$ thus pass second validation check also. Eavesdropper(if any) thus remains undetected if one uses $\mathbf{N}_4$ to generate raw key. However, when the same $\varrho$(satisfying Eq.(\ref{pecu8})) be used in $\mathcal{N}_4,$ the protocol gets aborted in the sifting stage as $\mathcal{C}_{\mathcal{N}_2}^{'}$ is violated. Based on all these observations, one thus gets the following result($\mathcal{R}_3,$say).\\
	$\mathcal{R}_3:$\textit{Second validation check is more stringent in $\mathcal{N}_4$ compared to that in $\mathbf{N}_4.$ }\\
	Combining $\mathcal{R}_1,\mathcal{R}_2$ and $\mathcal{R}_3,$ it can thus safely be concluded  that \textit{as an entanglement assisted network based QKD protocol, $\mathcal{N}_4$ offers better security than $\mathbf{N}_4.$}
	\subsection{Efficiency Attributable To Truly Connected Structure}
	From earlier discussions it is clear that validation of the entire protocol rests upon detection of full network nonlocal(FNN) correlations. Now, recalling the steps of $\mathcal{N}_4,$ it can be seen that using some copies of each of the three states, the protocol engages a global quantum state structure($\rho_G$,say) in form of $\rho_G$$=$$\otimes_{i=1}^3\rho_i$ over which the trusted parties perform local measurements(Steps.2(ii) and 2(b)). 
	It is thus this global structure($\rho_G$) of quantum states which is utilized in $\mathcal{N}_4$ to exploit some form of non-classical correlations(FNN) that cannot be decomposed into pairwise nonlocal resources in the protocol. So, here all the trusted parties need to collaborate together in order to generate a secret key. Intent of any eavesdropper to tamper with even one link of the network thus disrupts the entire network correlation pattern. This in turn increases the chance of detecting eavesdropper in $\mathcal{N}_4$.\\
	However, $\mathbf{N}_4$'s validation entirely relies upon the exploitation of bipartite nonlocality individually in each subset of the central and an edge party. Existence of pairwise nonlocality can be interpreted as a collection of independent correlations between different pairs. Hence, collaboration of all the legitimate users is not required for security purpose in $\mathbf{N}_4.$
	\section{Conclusion}\label{conc}
Manifestation of the notion of full network nonlocality for security analysis in a network based QKD protocol($\mathcal{N}_4$) has been the mainstay of present work. 
A trilocal star network based entanglement assisted protocol, having two layers of validation checks, has been designed for generating quantum key. While one such check step exploits violation of trilocal inequality, the other one relies upon reducing QBER below some threshold value($\mathbf{Q}_0$). Exploiting full network nonlocality for checking validation of the protocol has aided in characterizing arbitrary two-qubit states in context of utilizing them to execute $\mathcal{N}_4$ successfully. Such characterization has been obtained in terms of singular values of correlation tensors of the states used. 
	\par Another network based QKD protocol($\mathbf{N}_4$) has been introduced involving only only upon Bell-CHSH violation based check steps. A comparison of the security offered by $\mathcal{N}_4$ and $\mathbf{N}_4$ clearly points out that the former is more stringent in checking the validity of the protocol to be used for generating raw key. This in turn ensures efficiency of $\mathcal{N}_4$ over $\mathbf{N}_4$ from perspective of offering unconditional security. 
	\par Present work considers a particular trilocal inequality(Eq.(\ref{ineqs})) for the purpose of full network nonlocality detection. However, it will be interesting to use some other existing inequalities for checking full network nonlocality\cite{rnew}. Also it will be interesting to generalize the present approach for designing $n$-partite QKD protocol for any finite $n.$ Such a generalization warrants future investigation owing to extensive technological advancement towards development of scalable quantum networks. Also violation of the trilocal inequality serves only as a necessary criterion for verifying security. It will be interesting to frame both necessary and sufficient security criteria for $\mathcal{N}_4.$ Performing a systematic study of utilizing arbitrary-party and unbounded-input network inequalities in $n$-local networks\cite{pan} for secret key distribution also warrants future investigation.  
	\par In practical scenarios, testing of any correlator based inequality can suffer from several loopholes such as detection loopholes\cite{lp2}, locality loopholes\cite{lp11,lp1}, freedom-of-choice loophole\cite{lp3}. Also classical communication over public channel(for key generation) forms a potent factor of experimental imperfections. Resistance of $\mathcal{N}_4$ against detection loopholes has been analyzed here considering some specific strategy adopted by the trusted parties for retaining no-click events. Similarly, exploring possible means of closing other loopholes is a potential direction of future research. 
	\par $\mathcal{N}_4$ has been analyzed from Eve's perspective considering a specific adversarial strategy. Particularly, Eve holds a purification separately for each $\mathbf{S}_i.$	By doing so, Eve obtains the strongest possible side information individually from each source while still respecting trilocality assumption. However, such an adversarial strategy is not the most general one as Eve, controlling all the sources, may introduce cross-correlations in between the sources, thereby fetching more information from the network. Analyzing $\mathcal{N}_4$ in such unrestricted adversarial scenario is inevitably a significant direction of future research. 
\section*{Data Availability Statement}
No data is associated with this work.
		\section*{Research Funding}
	This research did not receive funding.
\section*{Competing Interests}
	The author has no competing interests as defined by Springer, or other interests that might be perceived to influence the results and/or discussion reported in this paper.
	\section*{Appendix.A}
	\textit{Proof of theorem.\ref{theo1}:} Here $\rho_1$$=$$\rho_2$$=$$\rho_3$$=$$\varrho$ with correlation tensor $T$$=$$\textmd{diag}(t_1,t_2,t_3).$\\
	It is clear from discussion in main text that finding critical error rate $\mathbf{Q}_0$ reduces to the task of finding maximum value of $\mathbf{H}$(Eq.(\ref{cr6})) with respect to $t_{i}$ subject to the constraint that trilocal inequality(Eq.(\ref{ineqs})) is not violated:
	\begin{eqnarray}\label{appa1}
		\textmd{\small{Maximize}}&&(2+t_1+t_2)^3\nonumber\\
		\textmd{\small{Sub To:}}&& t_{1}^2+t_{2}^2\leq 2^{\frac{2}{3}}
	\end{eqnarray}
	Consider the following maximization problem:
	\begin{eqnarray}\label{appa4}
		\textmd{\small{Maximize}}&&t_1+t_2\nonumber\\
		\textmd{\small{Sub To:}}&&t_1^2+t_2^2=2^{\frac{2}{3}}\nonumber\\
		&&	t_1,t_2\geq 0
	\end{eqnarray}
	Maxima obtained from above maximization problem(Eq.(\ref{appa4})) will be the maxima for the required maximization problem(Eq.(\ref{appa1})).\\
	Lagrangian($\mathcal{L}(t_1,t_2,\lambda),$say) corresponding to the maximization problem(Eq.(\ref{appa4})) is given by:
	\begin{eqnarray}\label{appa6}
		\mathcal{L}(t_1,t_2,\lambda)&=&t_1+t_2+\lambda(t_1^2+t_2^2-2^{\frac{2}{3}})
	\end{eqnarray}
	Critical points are given by:
	\begin{eqnarray}\label{appa7}
		\frac{\partial\mathcal{L}}{\partial t_j}&=&0\nonumber\\
		t_j&=&-\frac{1}{2\lambda}\,\,\forall j=1,2
	\end{eqnarray}
	Using values of $t_1,t_2$ from above in the constraint of Eq.(\ref{appa6}):
	\begin{eqnarray}
		\frac{1}{2\lambda^2}&=&2^{\frac{2}{3}}\nonumber\\
		\frac{1}{2\lambda}&=&\pm2^{\frac{1}{6}}
	\end{eqnarray}
	Using $\frac{1}{2\lambda}$$=$$-2^{\frac{1}{6}}$ in the constraint, one gets:
	\begin{eqnarray}
		t_j&=&2^{-\frac{1}{6}}\,\,\forall j=1,2
	\end{eqnarray}
	The critical point $K$(say) of the maximization problem(Eq.(\ref{appa4})) is thus given by:
	\begin{eqnarray}\label{appa8}
		K&=&(2^{-\frac{1}{6}},2^{-\frac{1}{6}})
	\end{eqnarray}
	At the critical point $K$(Eq.(\ref{appa8})), the Hessian matrix($\mathcal{H}$,say) is given by:
	\begin{eqnarray}
		[\mathcal{H}]_K&=&\left[ {\begin{array}{ccc}
				0&h_{t_1}&h_{t_2}\\
				h_{t_1}&\mathcal{L}_{t_1t_1}&\mathcal{L}_{t_1t_2}\\
				h_{t_2}&\mathcal{L}_{t_1t_2}&\mathcal{L}_{t_2t_2}\\
		\end{array} } \right]_{K}\nonumber\\
		&=&\left[ {\begin{array}{ccc}
				0&2^{\frac{5}{6}}&2^{\frac{5}{6}}\\
				2^{\frac{5}{6}}&2^{\frac{7}{6}}&0\\
				2^{\frac{5}{6}}&0&2^{\frac{7}{6}}\\
		\end{array} } \right]\nonumber\\
	\end{eqnarray}
	Determinant of $[\mathcal{H}]_K$ turns out to be negative. Hence $K$(Eq.(\ref{appa8})) is the maxima of the maximization problem given in Eq.(\ref{appa4})  and hence the maxima of the original maximization problem(Eq.(\ref{appa1})) .\\
	Maximum value of $\mathbf{H}$ is thus given by:
	\begin{eqnarray}
		\mathbf{H}]_{\textmd{\tiny{Max}}}&=&8(1+2^{-\frac{1}{6}})^3\nonumber\\
		&=&4\sqrt{2}(1+2^{\frac{1}{6}})^3
	\end{eqnarray}
	Minimum value of QBER $\mathbf{Q}_0$ is thus given by:
	\begin{eqnarray}
		\mathbf{Q}_0&=&1-\frac{\sqrt{2}(1+2^{\frac{1}{6}})^3}{16}\approxeq 0.154887
	\end{eqnarray}
	Hence the theorem is proved.   $\blacksquare$
	\section*{Appendix.B}
	\textit{Proof of theorem.\ref{theo11}:}Here the states $\rho_1,\rho_2,\rho_3$ are not all identical. \\
	Using expression of $\mathbf{H}$ provided by Eq.(\ref{cr6i}), the original maximization problem to be solved here is given by:
	\begin{eqnarray}\label{appna1}
		\textmd{\small{Maximize}}&&\sum_{\substack{j_i=1,2\\\forall i=1,2,3}}\Pi_{i=1}^3(1+t_{i,j_i})\nonumber\\
		\textmd{\small{Sub To:}}&& \Pi_{i=1}^3 (t_{i,1})^\frac{2}{3}+\Pi_{i=1}^3 (t_{i,2})^\frac{2}{3}\leq 2^{\frac{2}{3}}
	\end{eqnarray}
	Now $t_{i,j}$$\in$$[0,1]$ $\forall i,j.$\\
	The expression to be maximized here is the sum of the products of positive quantities $t_{i,j}.$ Clearly
	$t_{i,j}$$=$$1$ $\forall i,j$ will provide the maximum value. However, it will not respect the constraint above(Eq.(\ref{appna1})) provided by non-violation of the trilocal inequality(Eq.(\ref{ineqs})) as
	\begin{eqnarray}
		\Pi_{i=1}^3 (t_{i,1})^\frac{2}{3}+\Pi_{i=1}^3 (t_{i,2})^\frac{2}{3}=2\nleq 2^{\frac{2}{3}}
	\end{eqnarray}
	So,$t_{i,j}$$=$$1$ $\forall i,j$ cannot be the required maxima of the maximization problem(Eq.(\ref{appna1})).\\
	W.L.O.G., let $t_{i,j}$$=$$1$ $\forall i,j$ except $t_{3,2}.$ For that the given constraint in Eq.(\ref{appna1}) can be satisfied provided $t_{3,2}$ abides by the following constraint:
	\begin{eqnarray}\label{appna2}
		t_{3,2}^{\frac{2}{3}}&\leq & 2^{\frac{2}{3}}-1
	\end{eqnarray}
	The objective function in Eq.(\ref{appna1}) gets simplified as:
	\begin{eqnarray}
		\sum_{\substack{j_i=1,2\\\forall i=1,2,3}}\Pi_{i=1}^3(1+t_{i,j_i})&=&16(3+t_{3,2})	.
	\end{eqnarray}
	So the task reduces to solving the following maximization problem:
	\begin{eqnarray}\label{appna4}
		\textmd{\small{Maximize}}&&\,t_{3,2}\nonumber\\
		\textmd{\small{Sub To:}}&&\,	(t_{3,2})^{\frac{2}{3}}= 2^{\frac{2}{3}}-1
	\end{eqnarray}
	Lagrangian for above maximization problem(Eq.(\ref{appna4})) is given by:
	\begin{eqnarray}\label{appna5}
		\mathcal{L}(t_{3,2},\lambda)&=&t_{3,2}+\lambda(t_{3,2}^{\frac{2}{3}}-2^{\frac{2}{3}}+1)
	\end{eqnarray}
	Using the Lagrangian method, the maxima($t_{3,2}^{'}$,say) of the above maximization problem(Eq.(\ref{appna4})) is given by:
	\begin{eqnarray}\label{appna6}
		t_{3,2}^{'}&=&(2^{\frac{2}{3}}-1)^{\frac{3}{2}}.
	\end{eqnarray}
	Maximum value of $\mathbf{H}$ is thus given by:
	\begin{eqnarray}
		\mathbf{H}]_{\textmd{\tiny{Max}}}&=&16(3+(2^{\frac{2}{3}}-1)^{\frac{3}{2}})\nonumber
	\end{eqnarray}
	Minimum value of QBER $\mathbf{Q}_0$ is thus given by:
	\begin{eqnarray}\label{appna7}
		\mathbf{Q}_0&=&1-\frac{16(3+(2^{\frac{2}{3}}-1)^{\frac{3}{2}})}{2^6}\nonumber\\
		&=&1-\frac{3+(2^{\frac{2}{3}}-1)^{\frac{3}{2}}}{4} \approxeq 0.13745.
	\end{eqnarray}
	\section*{Appendix.C}
	\textit{Proof of theorem.\ref{theo2}:} Here finding critical error rate $\mathbf{Q}_0^{'}$ reduces to the task of finding maximum value of $\mathbf{H}$(Eq.(\ref{cr6})) with respect to $t_{i,j}$ subject to the following respective constraints:
	\begin{itemize}
		\item [(i)] $c$$=$$3:$ Bell-CHSH inequality is not violated by bipartite correlations in any of the 3 pairs $(A_1,A_2),$ $(A_1,A_3)$ and $(A_1,A_4)$. 
		\item [(ii)] $c$$=$$2:$ Bell-CHSH inequality is not violated by bipartite correlations in two of the 3 pairs $(A_1,A_2),$ $(A_1,A_3)$ and $(A_1,A_4)$.
		\item [(iii)] $c$$=$$1:$ Bell-CHSH inequality is not violated by bipartite correlations in only one of the 3 pairs $(A_1,A_2),$ $(A_1,A_3)$ and $(A_1,A_4)$.  
	\end{itemize} 
	\par\textit{Proof of (i):} Constraints applicable here are:
	\begin{eqnarray}\label{appb0}
		t_{i1}^2+t_{i2}^2&\leq& 1,\,\,\forall i=1,2,3
	\end{eqnarray}
	Original maximization problem in this case is:	
	\begin{eqnarray}\label{appb1}
		\textmd{\small{Maximize}}&&\sum_{\substack{j_i=1,2\\\forall i=1,2,3}}\Pi_{i=1}^3(1+t_{i,j_i})\nonumber\\
		\textmd{\small{Sub To:}}&& 	t_{i1}^2+t_{i2}^2\leq 1,\,\,\forall i=1,2,3
	\end{eqnarray}
	Let $M_j$$=$$\textmd{Max}_{i=1}^3 t_{i,j},\,\,j$$=$$1,2.$\\
Then 
\begin{eqnarray}\label{unn1}
t_{ij}&\leq& M_j,\,\,\forall i=1,2,3,\,\forall j=1,2.
\end{eqnarray}
Expanding the objective function in Eq.(\ref{appb1}) gives:
\begin{eqnarray}\label{unn2}
\Pi_{i=1}^3(1+t_{i,1})&+&\Pi_{i=1}^3(1+t_{i,2})+\nonumber\\
(1+t_{1,1})(1+t_{2,1})(1+t_{3,2})&+&(1+t_{1,1})(1+t_{2,2})(1+t_{3,1})+\nonumber\\
(1+t_{1,1})(1+t_{2,2})(1+t_{3,2})&+&
(1+t_{1,2})(1+t_{2,1})(1+t_{3,1})+\nonumber\\	(1+t_{1,2})(1+t_{2,2})(1+t_{3,1})&+&(1+t_{1,2})(1+t_{2,1})(1+t_{3,2}).\nonumber\\
\end{eqnarray}
Using Eqs.(\ref{unn1},\ref{unn2}), objective function in Eq.(\ref{appb1}) gets upper bounded as follows:
\begin{eqnarray}
\sum_{\substack{j_i=1,2\\\forall i=1,2,3}}\Pi_{i=1}^3(1+t_{i,j_i})&\leq&(1+M_1)^3+(1+M_2)^3+\nonumber\\
3(1+M_1)^2(1+M_2)&+&3(1+M_1)(1+M_2)^2\nonumber\\
&=&(2+M_1+M_2)^3.
\end{eqnarray}
Consider the following maximization problem:
\begin{eqnarray}\label{appbs4}
	\textmd{\small{Maximize}}&&(2+M_1+M_2)^3\nonumber\\
	\textmd{\small{Sub To:}}&&M_1^2+M_2^2=1\nonumber\\
	&&	M_1,M_2\geq 0
\end{eqnarray}
Maxima obtained from above maximization problem(Eq.(\ref{appbs4})) will be the maxima for the required maximization problem(Eq.(\ref{appb1})) provided the maxima satisfy the original constraint(Eq.(\ref{appb0})).\\
Again, $M_1,M_2$ being both positive, optimal solution for maximization problem(Eq.(\ref{appbs4})) will be provided by the optimal solution of the following maximization problem:
	\begin{eqnarray}\label{appb4}
		\textmd{\small{Maximize}}&&M_1+M_2\nonumber\\
		\textmd{\small{Sub To:}}&&M_1^2+M_2^2=1\nonumber\\
		&&	M_1,M_2\geq 0
	\end{eqnarray}
	Lagrangian($\mathcal{L}_1(M_1,M_2,\lambda),$say) corresponding to the maximization problem(Eq.(\ref{appb4})) is given by:
	\begin{eqnarray}\label{appb6}
		\mathcal{L}_1(M_1,M_2,\lambda)&=&M_1+M_2+\lambda(M_1^2+M_2^2-1)
	\end{eqnarray}
	Using the Lagrangian, the maxima($K_1$,say) of the above maximization problem(Eq.(\ref{appb4})) is given by:
	\begin{eqnarray}\label{appb8}
		K_1&=&(2^{-\frac{1}{2}},2^{-\frac{1}{2}})
	\end{eqnarray}
	$K_1$ satisfies the original constraint(Eq.\ref{appb0}). So, $K_1$ is the maxima of the original maximization problem(Eq.(\ref{appb1})) also.\\
	Maximum value of $\mathbf{H}$(Eq.\ref{cr5}) is thus given by:
	\begin{eqnarray}
		\mathbf{H}]_{\textmd{\tiny{Max}}}&=&8(1+2^{-\frac{1}{2}})^3\nonumber\\
		&=&2\sqrt{2}(1+\sqrt{2})^3
	\end{eqnarray}
	Minimum value of QBER $\mathbf{Q}_0^{'}$ in case.(i) is thus given by:
	\begin{eqnarray}\label{appb9}
		\mathbf{Q}_0^{'}&=&1-\frac{\sqrt{2}(1+\sqrt{2})^3}{32}\approxeq 0.37814
	\end{eqnarray}
	\par\textit{Proof of (ii):} W.L.O.G., let Bell-CHSH be violated by bipartite correlations arising from local measurements of $A_1$ and $A_2$ only, i.e., let $\rho_1$ only show Bell-CHSH violation. So,
	\begin{eqnarray}
		1< t_{1,1}^2+t_{1,2}^2&\leq& 2
	\end{eqnarray}
	
	Here $\mathbf{H}$ thus becomes upper bounded as follows:
	\begin{eqnarray}\label{appb10}
		\mathbf{H}&=& \sum_{\substack{j_i=1,2\\\forall i=1,2,3}}\Pi_{i=1}^3(1+t_{i,j_i})\nonumber\\
		&\leq& 4 \sum_{\substack{j_i=1,2\\\forall i=2,3}}\Pi_{i=2}^3(1+t_{i,j_i})
	\end{eqnarray}
	Equality in last line in Eq.(\ref{appb10}) holds when $t_{1,1}$$=$$t_{1,2}$$=$$1,$ i.e., when $\rho_1$ shows maximum Bell-CHSH violation.\\
	Constraints to be satisfied by $t_{2,1},t_{2,2},t_{3,1}$ and $t_{3,2}$ are:
	\begin{eqnarray}\label{appb11}
		t_{i,1}^2+t_{i,2}^2&\leq& 1,\,\,\forall i=2,3.
	\end{eqnarray}
	So the original maximization problem in this case is:
	\begin{eqnarray}\label{appb110}
		\textmd{\small{Maximize}}&&\sum_{\substack{j_i=1,2\\\forall i=2,3}}\Pi_{i=2}^3(1+t_{i,j_i})\nonumber\\
		\textmd{\small{Sub To:}}&& 	t_{i1}^2+t_{i2}^2\leq 1,\,\,\forall i=2,3
	\end{eqnarray}
	Let $\textmd{Max}_{i=2}^3 t_{i,j}$$=$$N_j,$ $j$$=$$1,2.$ \\
	Then,
	\begin{eqnarray}\label{appb12}
		\sum_{\substack{j_i=1,2\\\forall i=2,3}}\Pi_{i=2}^3(1+t_{i,j_i})&\leq& (2+N_1+N_2)^2
	\end{eqnarray}
	Also, 
	\begin{eqnarray}
		t_{i,1}^2+t_{i,2}^2&\leq& N_1^2+N_2^2 ,\,\,\forall i=2,3.	
	\end{eqnarray}
	Approaching as in (i), the maximization problem to be solved here is of the form:
	\begin{eqnarray}\label{appb13}
		\textmd{\small{Maximize}}&&N_1+N_2\nonumber\\
		\textmd{\small{Sub To:}}&&N_1^2+N_2^2=1\nonumber\\
		&&	N_1,N_2\geq 0
	\end{eqnarray}
	It may be noted that above maximization problem is same as that of the maximization problem in previous case. So, the maxima($K_2$,say) of the above maximization problem(Eq.(\ref{appb13})) is given by:
	\begin{eqnarray}\label{appb14}
		K_2&=&K_1=(2^{-\frac{1}{2}},2^{-\frac{1}{2}})
	\end{eqnarray}
	$K_2$ satisfies the original constraint(Eq.\ref{appb12}). So, $K_2$ is the maxima of the original maximization problem(Eq.(\ref{appb110})) also.\\
	Here maximum value of $\mathbf{H}$ is thus given by:
	\begin{eqnarray}
		\mathbf{H}]_{\textmd{\tiny{Max}}}&=&16(1+2^{-\frac{1}{2}})^2\nonumber\\
		&=&8(1+\sqrt{2})^2
	\end{eqnarray}
	Minimum value of QBER $\mathbf{Q}_0^{'}$ in case.(ii) is thus given by:
	\begin{eqnarray}\label{appb15}
		\mathbf{Q}_0^{'}&=&1-\frac{(1+\sqrt{2})^2}{8}\approxeq 0.27145.
	\end{eqnarray}

	\par\textit{Proof of (iii):} W.L.O.G., let Bell-CHSH be violated by bipartite correlations arising from local measurements of $A_1$ and $A_2$ and also from that of $A_1$ and $A_3$ , i.e., let each $\rho_1,\rho_2$ show Bell-CHSH violation.\\
	Hence,
	\begin{eqnarray}
		1< t_{i,1}^2+t_{i,2}^2&\leq& 2\,\,\forall i=1,2.
	\end{eqnarray}
	Here $\mathbf{H}$ thus becomes upper bounded as follows:
	\begin{eqnarray}\label{appb17}
		\mathbf{H}&\leq& 16\sum_{j=1}^2(1+t_{3,j})\nonumber\\
		&=&16(2+t_{3,1}+t_{3,2})
	\end{eqnarray}
	Equality in last line in Eq.(\ref{appb17}) holds when $t_{i,1}$$=$$t_{i,2}$$=$$1\forall i$$=$$1,2$ i.e., when $\rho_1,\rho_2$ both show maximum violation.\\
	Constraints to be satisfied by $t_{3,1}$ and $t_{3,2}$ is:
	\begin{eqnarray}\label{appb18}
		t_{3,1}^2+t_{3,2}^2&\leq& 1.
	\end{eqnarray}
	So the original maximization problem in this case is:
	\begin{eqnarray}\label{appb19}
		\textmd{\small{Maximize }}&& 2+t_{3,1}+t_{3,2}\nonumber\\
		\textmd{\small{Sub To: }}&& 	t_{31}^2+t_{32}^2\leq 1\nonumber\\
		&& t_{3,1},t_{3,2}\geq 0.
	\end{eqnarray}
	Proceeding as in last two cases, solution of above maximization problem(Eq.(\ref{appb19})) is given by:
	Here maximum value of $\mathbf{H}$ is thus given by:
	\begin{eqnarray}
		\mathbf{H}]_{\textmd{\tiny{Max}}}&=&16(2+\sqrt{2})
	\end{eqnarray}
	Minimum value of QBER $\mathbf{Q}_0^{'}$ in case.(iii) is thus given by:
	\begin{eqnarray}\label{appb20}
		\mathbf{Q}_0^{'}&=&1-\frac{2+\sqrt{2}}{4}\approxeq 0.14645.
	\end{eqnarray}
	Combining Eqs.(\ref{appb9},\ref{appb15},\ref{appb20}), one gets:\\
	$\mathbf{Q}_0^{'}$$=$$1-(\frac{1+\sqrt{2}}{2\sqrt{2}})^c$ where $c$($1$$\leq$$c$$\leq$$3$) denote number of pairs of central and extreme parties that do not observe Bell-CHSH violation. \\
	This completes the proof of the theorem.  $\blacksquare$
	
	\end{document}